\documentclass[a4paper]{article}
\usepackage{color}
\usepackage{amsmath}
\usepackage{amsthm}
\usepackage{amssymb}
\usepackage{amsmath}
\usepackage{mathcomp}
\usepackage{dsfont}
\usepackage[toc,page]{appendix}
\usepackage[T1]{fontenc}

\newtheorem{definition}{Definition}[section]
\newtheorem{lemma}[definition]{Lemma}

\newtheorem{proposition}[definition]{Proposition}
\newtheorem{theorem}[definition]{Theorem}
\newtheorem{remark}[definition]{Remark}

\numberwithin{equation}{section}

\def\cO{\mathcal{O}}
\def\cU{\mathcal{U}}
\def\eps{\varepsilon}
\def\bR{\mathbb{R}}
\def\bN{\mathbb{N}}
\def\wt{\widetilde}

\def\tr{\mathrm{tr}}

\title{From Hartree dynamics to the relativistic Vlasov equation}

\author{Elia Dietler, Simone Rademacher, Benjamin Schlein \\
\\
Institute for Mathematics, University of Zurich \\ Winterthurerstrasse 190, 
8057 Zurich, Switzerland}

\begin{document}

 \maketitle

\begin{abstract}
We derive the relativistic Vlasov equation from quantum Hartree dynamics for fermions with relativistic dispersion in the mean-field scaling, which is naturally linked with an effective semiclassic limit. 
Similar results in the non-relativistic setting have been recently obtained in \cite{BPSS}. The new 
challenge that we have to face here, in the relativistic setting, consists in controlling the difference between the quantum kinetic energy and the relativistic transport term appearing in the Vlasov equation. 
\end{abstract}

\section{Introduction}
In this paper, we consider the time evolution of large systems of weakly interacting fermions with a relativistic dispersion law. In particular, we are interested in systems of fermions in the mean field regime. Particles are initially confined in a volume of order one and interact through a potential varying on length scales of order one (so that each particles effectively interact with all other particles). 

Because of the Pauli principle, the kinetic energy of the particles is typically of order $N^{4/3}$. To make sure that the competition between kinetic and potential energy is non-trivial, the coupling constant in front of the interaction potential should be small, proportional to $N^{-2/3}$. Writing $\eps = N^{-1/3}$ and appropriately scaling the mass of the fermions, we end up with the many-body Schr\"odinger equation
\begin{equation}\label{eq:schrN} i\partial_t \psi_{N,t}  = \left[ \sum_{j=1}^N \sqrt{1 - \eps^2 \Delta_{x_j}} + \frac{1}{N} \sum_{i<j} V(x_i -x_j) \right] \psi_{N,t} \end{equation}
for the $N$-particle wave function $\psi_{N,t} \in L^2_a (\bR^{3N})$, the subspace of $L^2 (\bR^{3N})$ consisting of functions that are antisymmetric with respect to permutations. Remark that the mean-field regime is linked with a semiclassical limit, with $\eps = N^{-1/3}$ playing the role of Planck's constant (the notation $\eps = N^{-1/3}$ will be used throughout the rest of the paper). 

In \cite{BPS2} it was shown, extending results obtained in the non-relativistic setting in \cite{EESY,BPS1} and later also in \cite{BJPSS,PP,BBPPT}, that the reduced density $\gamma^{(1)}_{N,t} = N \tr_{2,\dots ,N} 
|\psi_{N,t} \rangle \langle \psi_{N,t} |$ associated with the solution of (\ref{eq:schrN}) can be approximated, for initial data close to a Slater determinant with reduced density $\omega_{N}$ satisfying certain natural semiclassical commutator bounds, by the solution of the relativistic Hartree equation 
\begin{equation}\label{eq:HF} i\eps \partial_t \omega_{N,t} = \left[ \sqrt{1-\eps^2 \Delta}  + (V * \rho_t) , \omega_{N,t} \right] \end{equation}
with initial data $\omega_{N,0} = \omega_N$. Here we defined $\rho_t (x) = N^{-1} \omega_{N,t} (x;x)$. At time $t=0$, $\omega_{N,0}$ is the orthogonal projection into the $N$-dimensional subspace of $L^2 (\bR^3)$ spanned by the orbitals building the initial Slater determinant. It is then simple to check that, for all $t \in \bR$, $\omega_{N,t}$ is an orthogonal projection with rank $N$ and therefore the reduced density of a new, evolved, Slater determinant. We normalize here reduced densities so that $\tr \, \gamma_{N,t}^{(1)} = \tr \, \omega_{N,t} = N$ for all $t \in \bR$. 

Through $\eps = N^{-1/3}$, the Hartree equation (\ref{eq:HF}) still depends on $N$. To understand what happens in the limit as $N \to \infty$, we introduce the Wigner transform of $\omega_{N,t}$, defined as a function of the position $x \in \bR^3$ and of the velocity $v \in \bR^3$ through  
\begin{equation}\label{eq:wign} W_{N,t} (x;v) = \frac{\eps^3}{(2\pi)^3} \int dy \, \omega_{N,t} \left( x+ \frac{\eps y}{2} ; x - \frac{\eps y}{2}\right) e^{-i y \cdot v} \end{equation}
and normalized so that 
\[ \int dx dv \, W_{N,t} (x;v) = \eps^3 \, \tr \, \omega_{N,t} = 1.\] 
Observe that, with the Wigner transform $W_{N,t}$ one can reconstruct the reduced density $\omega_{N,t}$ by Weyl quantization, i.e. we find 
\begin{equation}\label{eq:weyl1} \omega_{N,t} (x;y) = N \int dv \, W_{N,t} \left( \frac{x+y}{2} , v \right) e^{i v \cdot \frac{(x-y)}{\eps}} \end{equation} 
Notice moreover that $W_{N,t}$ is defined so that $\int W_{N,t} (x;v) dv = N^{-1} \omega_{N,t} (x;x)$ is the density of particles localized close to $x$ and that, similarly, $\int W_{N,t} (x;v) dx$ is the density of particles with velocity close to $v$. On the other hand, since $W_{N,t}$ is typically not positive, it cannot be interpreted as a density of particles on phase space (this can be seen as an expression of Heisenberg's uncertainty principle). 

With (\ref{eq:HF}), we can derive an equation for the evolution of the Wigner transform $W_{N,t}$. We find:
\begin{equation}\label{eq:idW}
\begin{split}
i \varepsilon \partial_t W_{N,t}(x,v) = \; & \frac{\varepsilon^3}{(2 \pi)^3} \int d y \;  i \varepsilon \partial_t \omega_{N,t} \left( x+ \varepsilon y/2, x- \varepsilon y/2\right) e^{- i v \cdot y}
\\ = \; &\frac{\varepsilon^3}{(2 \pi)^3} \int d y \; \left[ \sqrt{1-\eps^2 \Delta} , \omega_{N,t} \right] \left( x+ \varepsilon y/2, x- \varepsilon y/2\right) e^{- i v \cdot y} \\
&+ \frac{\varepsilon^3}{(2 \pi)^3} \int d y \; \left[ (V*\rho_t) (x+ \eps y/2) - (V*\rho_t) (x-\eps y/2) \right] \\ &\hspace{5cm} \times \omega_{N,t} \left( x+ \varepsilon y/2, x- \varepsilon y/2\right) e^{- i v \cdot y} 
\end{split} 
\end{equation}

It is convenient to express the first term on the r.h.s. of (\ref{eq:idW}) in momentum space. Denoting by
\[ \widehat{\omega}_{N,t} (p;q) = \frac{1}{(2\pi)^3} \int \omega_{N,t} (x;y) e^{-ix \cdot p} e^{i y \cdot q} dx dy \]
the integral kernel of the operator $\omega_{N,t}$ in momentum space, we can write
\[ \begin{split} 
W_{N,t} (x;v) = \; &\frac{\eps^3}{(2\pi)^3} \int dy \; \omega_{N,t} (x+\eps y/2 ; x-\eps y/2) e^{-i y \cdot v} \\
= \; &\frac{\eps^3}{(2\pi)^3} \int dy  \frac{1}{(2\pi)^3} \int dp dq \; \widehat{\omega}_{N,t} (p;q) e^{ip \cdot (x+\eps y/2)} e^{-iq \cdot (x-\eps y/2)} e^{-iy \cdot v} \\ 
= \; &\frac{\eps^3}{(2\pi)^3} \int dp dq \; \widehat{\omega}_{N,t} (p;q) e^{i (p-q) \cdot x} \delta (v - \eps (p+q)/2) \\
= \; & \int dP \;  \widehat{\omega}_{N,t} \left( \frac{v}{\eps} + \frac{P}{2} ; \frac{v}{\eps} - \frac{P}{2} \right) e^{i P \cdot x} \end{split} \]
Hence, we find
\[ \begin{split} 
\frac{\eps^3}{(2\pi)^3} \int dy \, &\left[ \sqrt{1-\Delta} , \omega_{N,t} \right] (x+\eps y/2 ; x-\eps y/2) e^{-i v\cdot y} \\ =\; & \frac{\eps^3}{(2\pi)^6} \int dy dp dq \left( \sqrt{1+\eps^2 p^2} - \sqrt{1+\eps^2 q^2} \right) \widehat{\omega}_{N,t} (p;q) e^{ip \cdot (x+\eps y/2)} e^{-iq \cdot (x-\eps y/2)} e^{-iv\cdot y} \\
= \; & \frac{\eps^3}{(2\pi)^6} \int dp dq \left( \sqrt{1+\eps^2 p^2} - \sqrt{1+\eps^2 q^2} \right) \widehat{\omega}_{N,t} (p;q) e^{ix \cdot (p-q)} e^{-i y  \cdot (v - \eps (p+q)/2)} \\
=\;& \frac{1}{(2\pi)^3} \int dP \, \left( \sqrt{1+ \eps^2 (P/2 + v/\eps)^2} - \sqrt{1+\eps^2 (-P/2 + v/\eps)^2} \right) \widehat{\omega}_{N,t} (P/2 + v/\eps ; - P/2 + v /\eps) e^{iP \cdot x} \end{split} \]
Since
\[ \sqrt{1+\eps^2 (P/2 - v/\eps)^2} - \sqrt{1 + \eps^2 (P/2 + v/\eps)^2} = \frac{\eps P \cdot v}{\sqrt{1+v^2}} + \cO (\eps^2) \]
we expect that
\[ \begin{split} \frac{\eps^3}{(2\pi)^3} \int dy &\left[ \sqrt{1-\eps^2 \Delta} , \omega_{N,t} \right] (x+\eps y/2 ; x-\eps y/2) e^{iy \cdot v} \\ \simeq \; & \eps \int dP \frac{P \cdot v}{\sqrt{1+v^2}} \widehat{\omega}_{N,t} (P/2 + v/\eps ; -P/2 + v/\eps ) e^{i P \cdot x} \\ = \; &-i \eps \frac{v}{\sqrt{1+v^2}} \cdot \nabla_x W_{N,t} (x,v) \end{split} \]
in the limit $N \to \infty$ (and thus $\eps \to 0$). 
Similarly (but staying this time in position space), we may expect that the second term on the r.h.s. of (\ref{eq:idW}) can be approximated by
\[ \begin{split} \frac{\eps^3}{(2\pi)^3} \int dy & \,\left[ (V*\rho_t) (x+\eps y/2) - (V*\rho_t) (x-\eps y/2) \right] \omega_{N,t} (x+\eps y /2 ; x - \eps y/2) e^{-i v\cdot y} \\ & \simeq \frac{\eps^3}{(2\pi)^3} \int dy \nabla (V*\rho_t) (x)  \, \eps y \, \omega_{N,t} (x+\eps y/2 ; x-\eps y/2) e^{-i v\cdot y} = i \eps \nabla (V*\rho_t) (x) \cdot \nabla_v W_{N,t} (x,v) \end{split} \]
This heuristic computation suggests that, in the limit of large $N$, the Wigner transform $W_{N,t}$ associated with the solution of the relativistic Hartree equation (\ref{eq:HF}) converges to a limit $W_t$ satisfying the relativistic Vlasov equation
\begin{equation}\label{eq:rel-vlas} \partial_t W_t + \frac{v}{\sqrt{1+v^2}} \cdot \nabla_x W_t + \nabla (V*\rho_t) \cdot \nabla_v W_t = 0 \end{equation}
with $\rho_t (x) = \int W_t (x,v) dv$. 

The goal of this paper consists in deriving rigorous bounds establishing the convergence of the Hartree dynamics governed by (\ref{eq:HF}) towards the relativistic Vlasov equation (\ref{eq:rel-vlas}). Similar results have been recently obtained in \cite{BPSS} in the non-relativistic setting. Previous works concerning convergence towards the non-relativistic Vlasov equation include \cite{LP,AKN,AKN2,APPP,PP09} and \cite{NS,Sp}, where convergence was established directly starting from the many-body quantum 
evolution. The new challenge that we have to face here, with respect to the non-relativistic case, is the fact that the quantum kinetic energy gives a contribution to the evolution of the Wigner transform (the first term on the r.h.s. of (\ref{eq:idW})) that only approaches the transport term in the Vlasov equation (\ref{eq:rel-vlas}) in the limit $\eps \to 0$ (in the non-relativistic case, the correspondence between the commutator with the Laplace operator at the quantum level and the transport term in the Vlasov equation is exact, for all $\eps > 0$).


\section{Main results}
\label{sec:main}

To state our results, we define some useful norms. For a function $f$ defined on the phase space $(x,v) \in \bR^3 \times \bR^3$ and for $s \in \bN$, we define the Sobolev norm 
\[ \| f \|^2_{H^s}  = \sum_{|\beta| \leq s} \int |\nabla^\beta f (x,v)|^2 dx dv \]
where $\beta = (\beta_1, \dots , \beta_6) \in \bN^6$ is a multi-index and $|\beta| = \sum_{j=1}^6 \beta_j$. For 
$a \in \bN$, we will also use the weighted norms 
\[ \| f \|^2_{H^s_a} = \sum_{|\beta| \leq s} \int (1+x^2 +v^2)^a \, |\nabla^\beta f (x,v)|^2 \]

In our first theorem, we show convergence in trace norm, under strong assumptions on the regularity of the initial data. Recall that $\eps = N^{-1/3}$ throughout the paper. 
\begin{theorem}
\label{thm:trace}
 Let $V \in W^{2,\infty}\left(\mathbb{R}^3\right)$. Let $\omega_N$ be a sequence of reduced densities on $L^2\left(\mathbb{R}^3\right)$, with $\tr \, \omega_N = N, \; 0 \leq \omega_N\leq 1$ and with Wigner transform $W_N$ satisfying $\|W_N\|_{H^6_2}\leq C$, uniformly in $N$. We denote by $\omega_{N,t}$ the solution of the relativistic Hartree equation
 \begin{equation}
 \label{eq:Hartree}
 i \varepsilon \, \partial_t \omega_{N,t}= \left[ \sqrt{1-\varepsilon^2\Delta} + \left( V * \rho_t\right), \omega_{N,t} \right] 
 \end{equation}
 with $\rho_t(x) = N^{-1} \omega_{N,t}(x;x)$ and initial data $\omega_N$. On the other hand, we denote by $\widetilde{W}_{N,t}$ the solution of the Vlasov equation 
\begin{align}
\label{eq:Vlasov}
\partial_t\widetilde{W}_{N,t}+  \frac{v}{\sqrt{1+v^2}} \cdot \nabla_x \widetilde{W}_{N,t}= \nabla\left( V * \widetilde{\rho}_t\right) \cdot \nabla_v \widetilde{W}_{N,t}
\end{align} 
with $\widetilde{\rho}_t(x)=\int d v \; \widetilde{W}_{N,t}(x,v)$ and with initial data $\widetilde{W}_{N,0}= W_N$. Moreover, let $\widetilde{w}_{N,t}$ be the Weyl quantization of $\widetilde{W}_{N,t}$ defined as in \eqref{eq:weyl1}. Then, there exists a constant $C > 0$ (depending on $\|V\|_{W^{2,\infty}}$ and on $\sup_N \| W_N \|_{H^2_4}$, but not on the higher Sobolev norms of $W_N$) such that
\begin{equation}
\label{eq:trace}
\tr \, |\omega_{N,t}-\widetilde{\omega}_{N,t}| \leq C N \eps  \exp( C \exp(C|t|)) \left[ 1+ \sum_{k=1}^4 \varepsilon^k \sup_{N} \| W_N \|_{H_2^{k+2}} \right].
\end{equation}
\end{theorem}

\begin{remark} In the non-relativistic setting, the well-posedness of the Vlasov equation has been established in \cite{D}. As shown in \cite[Appendix A]{BPSS}, the proof of \cite{D} can be generalized to complex valued initial data. It is also easy to check that the same arguments can be used for the relativistic Vlasov equation (\ref{eq:Vlasov}).
\end{remark}

In the next theorem, we relax partly the regularity assumptions on the initial data. In contrast to Theorem \ref{thm:trace}, we only obtain bounds for the Hilbert-Schmidt norm of the difference 
$\omega_{N,t}-\widetilde{\omega}_{N,t}$. 
\begin{theorem} \label{thm:HS1}
Let $V \in L^1\left( \mathbb{R}^3\right)$ be such that
\begin{align}
\label{eq:assthm2}
\int d p \; | \widehat{V}(p)| (1+|p|^3) < \infty.
\end{align}
Let $\omega_N$ be a sequence of reduced densities on $L^2\left( \mathbb{R}^3 \right)$, with $\tr \, \omega_N=N, \; 0 \leq \omega_N \leq 1$ and with Wigner transform $W_N$ satisfying $\|W_N\|_{H^2_2} \leq C$, uniformly in $N$. As in Theorem \ref{thm:trace}, we denote by $\omega_{N,t}$ the solution of the relativistic Hartree equation \eqref{eq:Hartree} with initial data $\omega_N$ and by $\widetilde{\omega}_{N,t}$ the Weyl quantization of the solution $\widetilde{W}_{N,t}$ of the Vlasov equation \ref{eq:Vlasov} with initial data $\widetilde{W}_{N,0}= W_N$. Then, there exists a constant $C>0$ depending only on $\sup_{N}\|W_N\|_{H^2_2}$ and on the integral \eqref{eq:assthm2}, such that
\begin{align*}
\| \omega_{N,t}- \widetilde{\omega}_{N,t}\|_{\mathrm{HS}}\leq C \sqrt{N}\varepsilon \exp(C \exp(C|t|)).
\end{align*}
\end{theorem}

{F}rom a slightly different perspective, it is also possible to assume convergence of the Wigner transform $W_N$ of the initial reduced density $\omega_N$ towards a classical probability density $W_0$ and to compare then the Wigner transform $W_{N,t}$ of the solution of the Hartree equation (\ref{eq:Hartree}) with initial data $\omega_N$ with the solution $W_t$ of the Vlasov equation (\ref{eq:Vlasov}), with initial data $W_0$. In this case, the rate of the convergence in the Hilbert-Schmidt norm is established in the next theorem, under the same regularity assumptions appearing in Theorem \ref{thm:HS1} above. 
\begin{theorem}\label{thm:HS2}
Let  $V \in L^1\left( \mathbb{R}^3\right)$ be such that \eqref{eq:assthm2} holds true. Let $\omega_N$ be a sequence of reduced densities on $L^2 \left( \mathbb{R}^3\right)$, with $\tr \; \omega_N=N, \; 0 \leq \omega_N \leq 1$ and with Wigner transform $W_N$ satisfying $\| W_N \|_{H^2_2}\leq C$, uniformly in $N$. Furthermore, let $W_0$ be a probability density on $\mathbb{R}^3 \times \mathbb{R}^3$ with $\|W_0\|_{H^2_2}< \infty$ and such that
\begin{align*}
\| W_N -W_0\|_1 \leq C \kappa_{N,1}, \quad \text{and} \quad  \|W_N-W_0\|_2 \leq C \kappa_{N,2}
\end{align*}
for sequences $\kappa_{N,1}, \kappa_{N,2} \geq 0$ with $\kappa_{N,j} \rightarrow 0$ as $N \rightarrow \infty$ for $j=1,2$. Let $\omega_{N,t}$ denote the solution of the relativistic Hartree equation \eqref{eq:Hartree} with initial data $\omega_{N}$ and let $W_{N,t}$ be its Wigner transform. On the other hand, let $W_t$ denote the solution of the Vlasov equation \[ \partial_t W_{t} +  \frac{v}{\sqrt{1+v^2}} \cdot \nabla_x W_{t}= \nabla\left( V * \rho_t\right) \cdot \nabla_v W_{t}, \] with $\rho_t (x) = \int  W_t (x,v) dv$ and with initial data $W_{t=0} = W_0$. Then we have
\begin{align}
\|W_{N,t}-W_t\| \leq C \varepsilon \exp(\exp(C|t|))+C (\kappa_{N,1}+\kappa_{N,2})\exp(C|t|).
\end{align}
\end{theorem}

In the next two theorems we relax the regularity conditions of the initial data. As a consequence, the next two theorems can be applied for initial data approximating ground states of confined systems. The price we have to pay to extend the class of admissible data is that we only show convergence after  testing against a semi classical observable. 
\begin{theorem} \label{thm:semi}
Let $V \in L^1 \left( \mathbb{R}^3\right) $ such that 
\begin{equation}
\label{eq:assthm4}
\int d p \; | \widehat{V} (p) | (1+|p|^4) < \infty .
\end{equation}
Let $\omega_N$ be a sequence of reduced densities on $L^2 \left( \mathbb{R}^3 \right) $, with $\tr \, \omega_N = N$, $0 \leq \omega_N \leq 1$, such that
\begin{equation}\label{eq:comms}
\tr \, | [x,\omega_N]| \leq C N \varepsilon, \qquad \tr \, |[\varepsilon\nabla, \omega_N]| \leq C N \varepsilon.
\end{equation}
We assume that the Wigner transform $W_N$ of $\omega_N$ is such that 
\[ \| W_N \|_{W^{1,1}} = \sum_{|\beta| \leq 1} \int d x \; d v \; | \nabla^\beta W_N(x,v) | \leq C \]
uniformly in $N$. Let $\omega_{N,t}$ be the solution of the Hartree equation \eqref{eq:Hartree} with initial data $\omega_N$. On the other hand, let $\widetilde{\omega}_{N,t}$ be the Weyl quantization of the solution $\widetilde{W}_{N,t}$ of the Vlasov equation \eqref{eq:Vlasov} with initial data $W_N$. Then, there exists a constant $C>0$, such that
\[ 
| \tr \, e^{i p \cdot x + q \cdot \varepsilon \nabla} (\omega_{N,t}-\widetilde{\omega}_{N,t}) | \leq CN\varepsilon (1+|p|+|q|)^2 \exp ( \exp (C |t|))
\] 
for all $p,q \in \mathbb{R}^3, \; t \in \mathbb{R}$.
 \end{theorem}

The trace of the difference $\omega_{N,t}- \widetilde{\omega}_{N,t}$ tested against an observable of the form $\exp (ip \cdot x + \eps \nabla \cdot q)$, as appearing in the theorem above, can be expressed in terms of  Wigner transforms. For $\omega_N$ an arbitrary fermionic density, we have 
\[ \begin{split} 
\tr \, e^{i p \cdot x + q \cdot \varepsilon \nabla} \omega_{N} =& \int d x \; e^{i q \cdot p}e^{i p \cdot x} \omega_{N}(x-\varepsilon q ;x)\\
 =& N \int d x d v \; W_N(x,v) e^{i p \cdot x} e^{i q \cdot v} = N \widehat{W}_N(p,q), \end{split} \] where $\widehat{W}_N(p,q)$ denotes the Fourier transformation of $W_N(x,y)$.
\begin{theorem}
\label{thm:semi2}
Let $V \in L^1\left( \mathbb{R}^3 \right)$ satisfy \eqref{eq:assthm4}. Let $\omega_N$ be a sequence of reduced densities on $L^2\left(\mathbb{R}^3\right)$, with $\tr\, \omega_N=N, \; 0 \leq \omega_N \leq 1$ and such that (\ref{eq:comms}) holds true. We assume that the Wigner transform $W_N$ of $\omega_N$ is such that $\|W_N\|_{W^{1,1}}\leq C$ uniformly in $N$. Furthermore, let $W_0 \in W^{1,1}\left( \mathbb{R}^3 \times \mathbb{R}^3 \right)$ be a probability density, such that
\[ 
 \| W_N -W_0\|_1 \leq \kappa_N
\]
 for a sequence $\kappa_N$ with $\kappa_N \rightarrow 0$ as $N \rightarrow \infty$. Let $\omega_{N,t}$ be the solution of the Hartree equation \eqref{eq:Hartree} with initial data $\omega_N$ and let $W_{N,t}$ be the Wigner transform of $\omega_{N,t}$. On the other hand, let $W_t$ denote the solution of initial data $W_0$. Then we have
\[ 
 \sup_{p,q} \frac{1}{(1+|p|+|q|)^2} \left\vert \widehat{W}_{N,t}(p,q) - \widehat{W}_t(p,q)\right\vert \leq C (\varepsilon+ \kappa_N) \exp( \exp ( C|t|))
 \]
 \end{theorem}

\section{Trace norm convergence for regular data}
 \label{sec:trace}

In this section we prove Theorem \ref{thm:trace}. The proof is an adaptation of the proof of \cite[Theorem 2.1]{BPSS}. First of all, we observe that the Weyl quantization (\ref{eq:weyl1}) associated with the solution $\wt{W}_{N,t}$ of the relativistic Vlasov equation (\ref{eq:Vlasov}) satisfies the equation
\[ i\eps \partial_t \wt{\omega}_{N,t} = A_{N,t} + B_{N,t} \]
where the operator $A_{N,t}$ has the momentum space kernel
\[ \widehat{A}_{N,t} (p;q) = \frac{\eps^2 (p-q) \cdot (p+q)}{2 \sqrt{1+\frac{\eps^2}{4} (p-q)^2}} \, \widehat{\wt{\omega}}_{N,t} (p;q) \]
while the operator $B_{N,t}$ has the position space kernel 
\[ B_{N,t} (x;y)  = (x-y) \cdot \nabla (V * \wt{\rho}_t) \left( \frac{x+y}{2} \right) \, \wt{\omega}_{N,t} (x;y) \]

To compare the solution $\omega_{N,t}$ of the relativistic Hartree equation (\ref{eq:Hartree}) with $\wt{\omega}_{N,t}$, it is convenient to switch to the interaction picture. To this end, we define the time-dependent relativistic Hartree Hamiltonian 
\[ 
h_H(t) = \sqrt{1- \varepsilon^2 \Delta} + \left(V * \rho_t\right)(x) ,
\]
and the corresponding two-parameter group of unitary transformations $\cU (t;s)$ satisfying 
\begin{equation}\label{eq:Uts} i\eps \, \partial_t \, \cU (t;s) = h_H (t) \, \cU (t;s) \end{equation}
with $\cU (s;s) = 1$ for all $s \in \bR$. Under the assumptions on the interaction $V$, the existence of the dynamics $\cU (t;s)$ can be shown by constructing first the dynamics in the interaction picture (the time-dependent generator $\exp (it \sqrt{1-\eps^2 \Delta}) \left(V * \rho_t\right)(x) \exp (-it \sqrt{1-\eps^2 \Delta})$ of the evolution $\mathcal{W} (t;s) = e^{it \sqrt{1-\eps^2 \Delta}} \cU (t;s) e^{-is \sqrt{1-\eps^2 \Delta}}$ is bounded and depends continuously on $t$; standard results guarantee the existence of $\mathcal{W} (t;s)$ and thus of $\cU (t;s) = e^{-it \sqrt{1-\eps^2 \Delta}}  \mathcal{W} (t;s) e^{is \sqrt{1-\eps^2 \Delta}}$). Then, we compute 
\begin{equation}\label{eq:tr0} \begin{split} i \eps \,\partial_t \, \cU^* (t;0) (\omega_{N,t} - \wt{\omega}_{N,t}) \, \cU (t;0) = \; & - \mathcal{U}^*(t;0) \left[ h_H(t), \omega_{N,t} - \widetilde{\omega}_{N,t} \right] \mathcal{U}(t;0)\\
	&+ \mathcal{U}^*(t;0) \left( \left[ h_H(t), \omega_{N,t}\right] - A_{N,t} - B_{N,t} \right) \mathcal{U}(t;0)\\
	=\; & \mathcal{U}^*(t;0) \left( \left[ \sqrt{1-\varepsilon^2 \Delta}  , \widetilde{\omega}_{N,t} \right] - A_{N,t} \right) \mathcal{U}(t;0)\\
	&+ \mathcal{U}^*(t;0) \left[ V * (\rho_t - \wt{\rho}_t) , \widetilde{\omega}_{N,t} \right] \cU (t;0) + \cU^* (t;0) \, C_{N,t} \, \cU (t;0) \end{split} \end{equation}
where the operator $C_{N,t}$ has the position space kernel
\[ C_{N,t} (x;y) = \left[ ( V * \widetilde{\rho}_t) (x) - (V * \widetilde{\rho}_t)(y) - \nabla (V * \widetilde{\rho}_t)\left( \frac{x+y}{2}\right) \cdot (x-y) \right] \widetilde{\omega}_{N,t}(x;y) 
\]
Since $\omega_{N,0} = \widetilde{\omega}_{N,0}= \omega_N$, we obtain integrating over time and taking the trace norm
\begin{equation}
\begin{split}
\tr \, | \omega_{N,t}- \widetilde{\omega}_{N,t}| = \; & \frac{1}{ \varepsilon} \int_0^t d s \; \tr \, |[ \sqrt{1-\varepsilon^2 \Delta} ,\widetilde{\omega}_{N,s} ] - A_{N,s} | \\
	&+\frac{1}{ \varepsilon} \int_0^t d s \; \tr \, |[ V * (\rho_s-\widetilde{\rho}_s), \widetilde{\omega}_{N,s} ] | 
	+\frac{1}{\varepsilon} \int_0^t d s \; \tr \, |C_{N,s} | \label{eq:tr_1}
\end{split}
\end{equation}
The two terms on the second line can be estimated exactly as in the non-relativistic setting. Proceeding as in \cite[(3.6)-(3.13)]{BPSS} and as in \cite[(3.14)-(3.16)]{BPSS} and using Proposition \ref{prop:regest} below to 
propagate the regularity of the solution of the relativistic Vlasov equation, we find 
\begin{equation} 
\label{eq:tr_bound1}
\begin{split}
\frac{1}{\varepsilon} \int_0^t d s \;  \tr \, | [ V * ( \widetilde{\rho}_s - \rho), \widetilde{\omega}_{N,s}]|  \leq \; &C \int_0^t d s \; e^{C |s|} \tr \, | \omega_{N,s}- \widetilde{\omega}_{N,s}| \\
	&+ C e^{C |t|} N \varepsilon \left[ \| W_N\|_{H^2_2}+ \varepsilon \| W_N \|_{H^3_2}+ \varepsilon^2 \| W_N\|_{H_2^4}\right], 
\end{split}
\end{equation}
and, similarly, 
\begin{align}
\tr \, |C_{N,s}| \leq C e^{C |s|} N \varepsilon^2 \left[ \|W_N\|_{H^2_2} + \varepsilon \|W_N\|_{H_2^3}+ \varepsilon^2  \| W_N\|_{H^4_2} + \varepsilon^3  \|W_N\|_{H_2^5}\right]. \label{eq:tr_bound2}
\end{align}
for a constant $C >0$ depending on $\| W_N \|_{H_2^2}$ but not on higher Sobolev norms of $W_N$. 

Next, we bound the first term on the r.h.s. of (\ref{eq:tr_1}). 
We notice that 
\begin{equation}\label{eq:D0} \begin{split} 
 \tr \; |[ \sqrt{1-\varepsilon^2 \Delta} \, , & \, \widetilde{\omega}_{N,s} ] - A_{N,s}| \leq C \sqrt{N} \| D_{N,s} \|_\text{HS} \end{split} \end{equation}
where we defined the operator 
\begin{equation}\label{eq:Ds} D_{N,s} = (1-\eps^2 \Delta) (1+x^2) \left( [ \sqrt{1-\eps^2 \Delta}, \wt{\omega}_{N,s} ] - A_{N,s} \right) \end{equation}
and we used the fact that
\[ \| ( x^2 +1)^{-1} (1 - \varepsilon^2 \Delta )^{-1} \|_{\mathrm{HS}} \leq C \sqrt{N} \, . \]

To estimate the Hilbert-Schmidt norm of (\ref{eq:Ds}), we write its momentum space kernel as
\[ \widehat{D}_{N,s} (p;q) = (1 + \eps^2 p^2) (1- \Delta_p) \, F (p;q) \, \widehat{\wt{\omega}}_{N,s} (p;q) \]
with 
\begin{equation}\label{eq:Fpq} F(p;q) = \sqrt{1+ \eps^2 p^2} - \sqrt{1+\eps^2 q^2} - \frac{\eps^2}{2} \frac{(p-q) \cdot (p+q)}{\sqrt{1 + \frac{\eps^2}{4} (p+q)^2}} \end{equation}
We decompose 
\begin{equation}\label{eq:Dpq} \begin{split}  \widehat{D}_{N,s} (p;q) = & \; (1+\eps^2 p^2)  \left[ F(p;q) - \Delta_p F (p;q) \right] \,  \widehat{\wt{\omega}}_{N,s} (p;q) \\ &- 2 (1+\eps^2 p^2) \nabla_p F(p;q) \cdot \nabla_p \widehat{\wt{\omega}}_{N,s} (p;q) \\ &-(1+\eps^2 p^2) F(p;q) (\Delta_p \widehat{\wt{\omega}}_{N,s} (p;q))  =: \sum_{j=1,2,3} \widehat{D}_{N,s}^{(j)} (p;q)  \end{split} \end{equation}

For any twice continuously differentiable function $f$ and for any $p,q \in \bR^3$, we can write 
\begin{equation*}
\begin{split} 
f(p)= & \; f \left( \frac{p+q}{2}+ \frac{p-q}{2}\right) \\
=& \; f\left(\frac{p+q}{2}\right) + \int_0^1 d\kappa \; \frac{p-q}{2} \cdot (\nabla f) \left(  \frac{p+q}{2}+ \kappa \frac{p-q}{2}\right) \\
	=& \; f\left(\frac{p+q}{2}\right) + \frac{p-q}{2} \cdot (\nabla f) \left(  \frac{p+q}{2}\right) \\
	&+ \sum_{i,j=1}^3 \int_0^1 d\kappa \; \int_0^\kappa dr \; ( \partial_{ij} f ) \left( \frac{p+q}{2} + r  \frac{p-q}{2} \right) \frac{(p-q)_i (p-q)_j}{4}\label{eq:id2}
\end{split} \end{equation*}
Choosing $f (p) = \sqrt{1+\eps^2 p^2}$ and then applying the same formula with $p$ and $q$ switched, we find
\begin{equation}\label{eq:Fpq2} \begin{split} 
F(p;q) = \; &\frac{1}{4} \int_0^1 d\kappa \int_0^\kappa dr \\ 
&\times \left\{ \eps^2 (p-q)^2 \, \left[ \frac{1}{\left(1 + \eps^2 \left( \frac{p+q}{2} + r \frac{p-q}{2} \right)^2 \right)^{1/2}} - \frac{1}{\left(1 + \eps^2 \left( \frac{p+q}{2}  - r \frac{p-q}{2} \right)^2 \right)^{1/2}} \right]  \right. \\
&\hspace{3cm}  \left. + \eps^4 \left[ \frac{\left( (p-q) \cdot \left( \frac{p+q}{2} + r \frac{p-q}{2} \right) \right)^2}{\left(1+ \eps^2 \left( \frac{p+q}{2} + r \frac{p-q}{2} \right)^2 \right)^{3/2}} - 
 \frac{\left( (p-q) \cdot \left(\frac{p+q}{2} - r \frac{p-q}{2} \right) \right)^2}{\left(1+ \eps^2 \left( \frac{p+q}{2} - r \frac{p-q}{2} \right)^2 \right)^{3/2}} \right]  \right\} \end{split} \end{equation}
From this expression, we obtain the bounds
\begin{equation}\label{eq:F-bds} \begin{split} 
(1 + \eps^2 p^2) |F (p;q)| & \leq C \eps^2 |p-q|^2 (1 + \eps^2 |p+q|^2)^{1/2} + C \eps^4 |p-q|^4   \\
(1+ \eps^2 p^2) |\nabla_p F(p;q)| & \leq C \eps^2 |p-q| (1 + \eps^2 |p+q|^2)^{1/2} + C \eps^5 |p-q|^4 \\
(1+ \eps^2 p^2) |\Delta_p F(p;q)| & \leq C \eps^2 (1+ \eps^2 |p+q|^2)^{1/2} + C \eps^6 |p-q|^4 
\end{split} \end{equation}

Using (\ref{eq:F-bds}), we can estimate
\begin{equation}\label{eq:D11} \begin{split} \| D^{(1)}_{N,s} \|_\text{HS}^2 = \; &\int dp dq \, | \widehat{D}_{N,s}^{(1)} (p;q)|^2 \\ \leq \; &C \eps^4 \int dp dq \left| (1 + (p-q)^2) (1 +\eps^2 (p+q)^2 )^{1/2} \, \widehat{\wt{\omega}}_{N,s} (p;q) \right|^2 \\ &+ C \eps^8 \int dp dq \, \left| (p-q)^4 \, \widehat{\wt{\omega}}_{N,s} (p;q) \right|^2 \end{split} \end{equation}
We compute
\[ \begin{split} (p-q)^4 \, \widehat{\wt{\omega}}_{N,s} (p;q) &= \frac{(p-q)^4}{(2\pi)^3} \int dx dy \, \wt{\omega}_{N,s} (x;y) e^{-i x \cdot p} e^{i y \cdot q} 
\\ 
&= \frac{N (p-q)^4}{(2\pi)^3} \int dx dy dv \, \wt{W}_{N,s} \left( \frac{x+y}{2} , v \right) e^{i v \cdot \frac{x-y}{\eps}}  e^{-i x \cdot p} e^{i y \cdot q} \\
&= (p-q)^4 \int dR \; \wt{W}_{N,s} \left(R , \frac{\eps (p+q)}{2} \right) e^{-i R \cdot (p-q)} \\
&=  \int dR \,  (\Delta^2_R \wt{W}_{N,s}) \left( R, \frac{\eps (p+q)}{2} \right) e^{-i R \cdot (p-q)} \end{split} \]
and, similarly, 
\[ \begin{split}
(1+ (p-q)^2) &(1 +\eps^2 (p+q)^2 )^{1/2} \,  \widehat{\wt{\omega}}_{N,s} (p;q) \\ &= \int dR \, (1 +\eps^2 (p+q)^2 )^{1/2} [ (1-\Delta_R) \wt{W}_{N,s} ] \left( R , \frac{\eps (p+q)}{2} \right) \, e^{-i R \cdot (p-q)} \end{split} \]
With (\ref{eq:D11}), changing variables to $v = \eps (p+q)/2$ and $w = p-q$, this implies that 
\begin{equation}\label{eq:D1f} \begin{split} 
\| D^{(1)}_{N,s} \|_\text{HS}^2 \leq \; &C N \eps^4 \int dx dv \, (1+ v^2) |(1- \Delta_x) \wt{W}_{N,s} (x,v)|^2  + C N  \eps^8 \int dx dv \, | \Delta_x^2 \wt{W}_{N,s} (x,v)|^2 \\ \leq \; &C N \eps^4 \| \wt{W}_{N,s} \|_{H^2_1}^2 + C N  \eps^8 \|\wt{W}_{N,s} \|_{H^4_0} \end{split} \end{equation}

The Hilbert-Schmidt norm of the 
second term on the r.h.s. of (\ref{eq:Dpq}) can be bounded by 
\[ \begin{split} \| D^{(2)}_{N,s} \|_\text{HS}^2 = \; &\int dp dq \, | \widehat{D}^{(2)}_{N,s} (p;q)|^2 \\ \leq \; &C \eps^4 \sum_{\ell,j =1}^3 \int dp dq \, \left| (1+\eps^2 (p+q)^2)^{1/2} (p-q)_\ell \, \partial_{p_j} \, \widehat{\wt{\omega}}_{N,s} (p;q) \right|^2 \\ &+ C \eps^{10} \sum_{j=1}^3 \int dp dq \, \left| (p-q)^4 \,  \partial_{p_j} \widehat{\wt{\omega}}_{N,s} (p;q) \right|^2 \end{split} \]
We have
\begin{equation}\label{eq:D2-1} \begin{split} 
(p-q)^4 \partial_{p_j} \widehat{\wt{\omega}}_{N,s} (p;q) = \; &\frac{-iN (p-q)^4}{(2\pi)^3} \int dx dy dv  \, x_j \,  \wt{W}_{N,s} \left( \frac{x+y}{2} , v \right) \, e^{-i x \cdot p} e^{i y \cdot q} e^{i v \cdot \frac{x-y}{\eps}}   \\
= \; &\frac{-iN (p-q)^4}{(2\pi)^3} \int dR dr dv \, 
\wt{W}_{N,s} \left( R_j + \frac{r_j}{2} \right) \, \wt{W}_{N,t} \left( R , v \right) \,  e^{-i R \cdot (p-q)} e^{-i r \cdot \left( \frac{p+q}{2} - \frac{v}{\eps} \right)} \\
= \; & -i (p-q)^4 \int dR \, e^{-i R \cdot (p-q)} \\ &\hspace{.3cm} \times  \left[ R_j \,  \wt{W}_{N,s} \left(R, \frac{\eps (p+q)}{2}\right) - \frac{i\eps}{2}   (\partial_{v_j} \wt{W}_{N,s}) \left( R, \frac{\eps (p+q)}{2} \right) \right] \\
= \; & \int dR \, e^{-i R \cdot (p-q)} \\ &\hspace{.3cm} \times \left\{  \Delta_R^2  \left[ -i R_j \wt{W}_{N,s} (R, \eps (p+q)/2) - \frac{\eps}{2}   (\partial_{v_j} \wt{W}_{N,s}) (R, \eps (p+q)/2) \right] \right\} 
\end{split} \end{equation}
where $\nabla_v$ indicates the derivative with respect to the velocity components of $\wt{W}_{N,t}$. Similarly, we find 
\begin{equation}\label{eq:D2-2} \begin{split} 
(1 + \eps^2 (p+q)^2)^{1/2}  &(p-q)_{\ell} \, \partial_{p_j} \widehat{\wt{\omega}}_{N,s} (p;q) \\ = \; &\int dR \, e^{-i R \cdot (p-q)} (1+ \eps^2 (p+q)^2)^{1/2} 
\\ & \hspace{.3cm} \times \left\{ \partial_{R_\ell} \left[ R_j \, \wt{W}_{N,s} \left( R, \frac{\eps (p+q)}{2} \right) - i\frac{\eps}{2} (\partial_{v_j} \wt{W}_{N,s}) \left( R,\frac{\eps (p+q)}{2} \right) \right] \right\} \end{split} \end{equation}
Hence  
\begin{equation}\label{eq:D2f} \begin{split}
\| D^{(2)}_{N,s} \|_\text{HS}^2 \leq \; &C N \eps^4 \sum_{\ell,j =1}^3 \int dx dv \, (1+ v^2) \left| \partial_{x_\ell} \left[ x_j \wt{W}_{N,s} (x,v) - \frac{i\eps}{2} (\partial_{v_j} \wt{W}_{N,s}) (x,v)  \right] \right|^2 
\\ &+ C N \eps^{10} \sum_{j=1}^3 \int dx dv \left| \Delta_x^2 \left[ -i x_j \wt{W}_{N,s} (x,v) - \frac{\eps}{2} (\partial_{v_j} \wt{W}_{N,s}) (x,v) \right] \right|^2 \\
\leq \; &C N \eps^4 \| \wt{W}_{N,s} \|_{H^1_2}^2 + C N \eps^6 \| \wt{W}_{N,s} \|^2_{H^2_1} + C N \eps^{10} \| \wt{W}_{N,s} \|_{H^4_1}^2 + C N \eps^{12} \| \wt{W}_{N,s} \|_{H^5_0} \end{split} \end{equation}

As for the Hilbert-Schmidt of the third term on the r.h.s. of (\ref{eq:Dpq}), (\ref{eq:F-bds}) implies that 
\begin{equation}\label{eq:D30} \begin{split} \| D_{N,s}^{(3)} \|_\text{HS}^2 \leq \; &C \eps^4 \int dp dq \, \left| (1+\eps^2 (p+q)^2)^{1/2}  \Delta_p \, \widehat{\wt{\omega}}_{N,s} (p;q) \right|^2 \\ &+ C \eps^{12} \int dp dq \, \left| (p-q)^4 \Delta_p \,  \widehat{\wt{\omega}}_{N,s} (p;q) \right|^4 \end{split} \end{equation}
Analogously to (\ref{eq:D2-1}) and (\ref{eq:D2-2}), we can write 
\[ \begin{split} 
(p-q)^4 &\Delta_p \widehat{\wt{\omega}}_{N,s} (p;q) \\ = \; &
\int dR \, e^{-i R \cdot (p-q)} \, \\ & \hspace{.2cm} \times  \left\{ \Delta^2_R \left[ R^2 \wt{W}_{N,s} \left(R, \frac{\eps (p+q)}{2} \right) -i \eps R \cdot (\nabla_v W_{N,s}) \left( R , \frac{\eps (p+q)}{2} \right) - \frac{\eps^2}{4} \Delta_v W_{N,s} \left( R, \frac{\eps (p+q)}{2} \right) \right] \right\} 
\end{split} \]
and
\[ \begin{split} 
(1+\eps^2 &(p+q)^2)^{1/2} \, \Delta_p \,  \widehat{\wt{\omega}}_{N,s} (p;q) \\ = \; &
\int dR \, e^{-i R \cdot (p-q)}   \, (1+ \eps^2 (p+q)^2)^{1/2} \\ &\hspace{.2cm} \times  \left[ R^2 \wt{W}_{N,s} \left(R, \frac{\eps (p+q)}{2} \right) -i \eps R \cdot (\nabla_v W_{N,s}) \left( R , \frac{\eps (p+q)}{2} \right) - \frac{\eps^2}{4} \Delta_v W_{N,s} \left( R, \frac{\eps (p+q)}{2} \right) \right]  \end{split} \]
Inserting in (\ref{eq:D30}), we conclude that
\[ \begin{split}  \| D^{(3)}_{N,s} \|_\text{HS}^2 \leq \; &C \eps^4 N \int dx dv \, (1+v^2) \left| x^2 \wt{W}_{N,s} (x,v) - i\eps x \cdot \nabla_v \wt{W}_{N,s} (x,v) - \frac{\eps^2}{4} \Delta_v \wt{W}_{N,s} (x,v) \right|^2 \\ &+ C \eps^{12} N \int dx dv \, \left| \Delta^2_x \left[ x^2 \wt{W}_{N,s} (x,v) - i \eps x \cdot \nabla_v \wt{W}_{N,s} (x,v) - \frac{\eps^2}{4} \Delta \wt{W}_{N,s} (x,v) \right] \right|^2 \\ \leq \; & C N \eps^4 \| \wt{W}_{N,s} \|_{H^0_2}^2 + C N \eps^6 \| \wt{W}_{N,s} \|_{H_2^1}^2 + C N \eps^8 \| \wt{W}_{N,s} \|_{H^2_1}^2 + C N \eps^{12} \| \wt{W}_{N,s} \|^2_{H^4_2} \\ &+ C N \eps^{14} \| \wt{W}_{N,s} \|_{H^5_1} + C \eps^{16} N \| \wt{W}_{N,s} \|_{H^6_0} \end{split} \] 

Combining the last bound with (\ref{eq:D1f}) and (\ref{eq:D2f}) we obtain (estimating for simplicity 
all weights with $(1+x^2 +v^2)^2$) 
\[ \begin{split} 
\| D_{N,s} \|_\text{HS}  \leq \; &C \sqrt{N} \eps^2 \| \wt{W}_{N,s} \|_{H^2_2} + C \sqrt{N}  \eps^4 \|\wt{W}_{N,s} \|_{H^4_2}  + C \sqrt{N} \eps^{6} \| \wt{W}_{N,s} \|_{H^5_2} + C\sqrt{N} \eps^{8}  
\| \wt{W}_{N,s} \|_{H^6_2} \\ 
\leq \; &C \, \sqrt{N} \, \eps^2 \,  e^{K|t|} \left\{ 
\| W_N \|_{H^2_2} + \eps^2 \| W_N \|_{H^4_2} + \eps^4 
\| W_N \|_{H^5_2} + \eps^6 \| W_N \|_{H^6_2} \right\}
\end{split} \]
where we used Proposition \ref{prop:regest} to propagate the regularity of the solution of the relativistic Vlasvo equation. Inserting in (\ref{eq:D0}), we conclude that
\[ 
\tr \; |[ \sqrt{1-\varepsilon^2 \Delta}  \widetilde{\omega}_{N,s} ] - A_{N,s}| \leq C N \eps^2 e^{K|t|} \left\{ 
\| W_N \|_{H^2_2} + \eps^2 \| W_N \|_{H^4_2} + \eps^4 
\| W_N \|_{H^5_2} + \eps^6 \| W_N \|_{H^6_2 } \right\}
 \]
Together with (\ref{eq:tr_bound1}) and (\ref{eq:tr_bound2}), we obtain from (\ref{eq:tr_1}) that 
\[ \tr \, |\omega_{N,t} - \wt{\omega}_{N,t} | \leq C \int_0^t e^{C|s|} \tr \, |\omega_{N,s} - \wt{\omega}_{N,s}| + C N \eps \, e^{K |t|} \, \sum_{j=0}^{4} \eps^j \| W_N \|_{H^{j+2}_2} \]
By Gronwall's Lemma, we arrive at
\[ \tr \, |\omega_{N,t} - \wt{\omega}_{N,t} | \leq C N \eps \, \exp (C \exp (C|t|)) \, \sum_{j=0}^4 \eps^j \| W_N \|_{H^{j+2}_2} \]
for a constant $C > 0$ that may depend on the $H^2_2$ norm of $W_N$, but not on its higher Sobolev norms. This concludes the proof of Theorem \ref{thm:trace}.

\section{Hilbert-Schmidt norm convergence}
 \label{sec:hs}

In this Section, we prove Theorem \ref{thm:HS1} and Theorem \ref{thm:HS2}. To this end, we proceed similarly as in \cite[Section 4]{BPSS}, i.e. we approximate the initial Wigner transform $W_N$ by a regularized data $W_N^k$, satisfying the assumptions of Theorem \ref{thm:trace}. Using the inital data $W_N^k$ and its Weyl quantization $\omega_N^k$, we construct the solution of the relativistic Vlasov equation $\wt{W}_{N,t}^k$ and of the relativistic Hartree equation $\omega_{N,t}^k$. With Theorem \ref{thm:trace}, we compare $\omega_{N,t}^k$ with $\wt{\omega}_{N,t}^k$, the Weyl quantization of $\wt{W}_{N,t}^k$. At the end, we remove the regularization, comparing $\omega_{N,t}^k$ with $\omega_{N,t}$ and, separately, $\wt{\omega}_{N,t}^k$ with $\wt{\omega}_{N,t}$. 

\bigskip

{\bf Regularization of initial data.} For $k >0$, we define 
\[ g_k(x,v):=\frac{k^3}{(2\pi)^3}	e^{-\frac{k}{2}(x^2+v^2)} \]
and  
\[ 
W_N^k(x,v)=\left( W_N* g_k \right)(x,v)=\int d x' \; d v' \; g_k(x-x',v-v')W_N(x',v').
\] 
It is simple to show that $W_N^k$ satisfies 
\begin{equation}\label{eq:reg_estimate} 
\begin{split} 
\|W_N^k\|_{H_2^j} &\leq C \|W_N\|_{H_2^2} \quad \text{if} \; \; j \leq 2, \; \; \text{and } \\ 
\|W_N^k\|_{H_2^j} &\leq C k^{\frac{j-2}{2}} \| W_N\|_{H_2^2} \quad \text{for} \; \; j=3,\dots, 6. 
\end{split} \end{equation}
Furthermore, we notice that 
\begin{align}
\label{eq:sobolevreg}
\| W_N - W_N^k\|_{H^s_a} \leq \frac{C}{\sqrt{k}} \| W_N \|_{H_a^{s+1}}
\end{align}
for $s=0,1$ and $a > 0$ (here we use the convention $H^0=L^2$). We denote by $\omega_N^k$ the Weyl quantization of $W_N^k$. {F}rom (\ref{eq:wign}) and (\ref{eq:weyl1}), we find 
\[ \omega_N^k (x;y) =  \frac{1}{(2 \pi)^3} \int d w \, d z  \; e^{-\frac{z^2}{2}} e^{-\frac{w^2}{2}} \left[ e^{i x \cdot \frac{w}{\varepsilon \sqrt{k}}} e^{\frac{z}{\sqrt{k}}\cdot\nabla} \, \omega_N \, e^{-\frac{z}{\sqrt{k}} \cdot \nabla} e^{-i x \cdot \frac{w}{\varepsilon\sqrt{k}}} \right] (x;y) \]
Thus, $\omega_{N}^k$ is a convex combination of fermionic densities and, hence, a fermionic density itself (i.e. $0 \leq \omega_N^k \leq 1$ and $\tr \, \omega_N^k=N$). From \eqref{eq:sobolevreg}, we find
 \begin{align}
 \| \omega_N - \omega_N^k \|_{\mathrm{HS}} = \sqrt{N} \| W_N - W_N^k \|_{L^2} \leq C \sqrt{\frac{N}{k}} \| W_N\|_{H^1}.\label{eq:hs_initial} 
 \end{align}

\bigskip

{\bf Notation.} We denote by $\omega_{N,t}$ and $\omega_{N,t}^k$ the solution of the Hartree equation with the initial data $\omega_N$ and, respectively $\omega_N^k$. On the other hand, $\widetilde{\omega}_{N,t}$ and $\widetilde{\omega}_{N,t}^k$ denote the Wigner transform of the solutions $\widetilde{W}_{N,t}$ and $\widetilde{W}_{N,t}^k$ of the Vlasov equation with initial data $W_N$ and respectively $W_N^k$.  Since the Vlasov equation preserves $L^p$-norms, we have 
 \begin{align*}
 \|\widetilde{\omega}_{N,t}\|_{\mathrm{HS}} = \sqrt{N} \, \| \widetilde{W}_{N,t}\|_2= \sqrt{N} \, \|W_N\|_2
 \end{align*}
 and similarly
 \begin{align*}
 \| \widetilde{\omega}_{N,t}^k\|_{\mathrm{HS}}= \sqrt{N} \, \| W_N^k \|_2
 \end{align*}
 for all $t \in \mathbb{R}$. 
 To prove Theorem \ref{thm:HS1} we need to compare $\omega_{N,t}$ with $\widetilde{\omega}_{N,t}$. To this end, we will first compare $\omega_{N,t}^k$ and $\widetilde{\omega}_{N,t}^k$. Afterwards, we will compare separately $\omega_{N,t}$ with $\omega_{N,t}^k$ and $\widetilde{\omega}_{N,t}$ with $\widetilde{\omega}_{N,t}^k$.
 
 \bigskip
 
{\bf Comparison of $\omega_{N,t}^k$ with $\widetilde{\omega}_{N,t}^k$. } We show here that 
\begin{equation}\label{eq:step1-HS} \| \omega_{N,t}^k - \wt{\omega}_{N,t}^k \|_\text{HS} \leq C N^{1/2} \eps \exp ( C \exp ( C |t|))  \left[ 1 + (\eps \sqrt{k})^4 \right] \end{equation}
for a constant $C$ depending on $\sup_N \| W_N \|_{H^2_2}$ (but not on higher Sobolev norms). 

To prove (\ref{eq:step1-HS}), we use Theorem \ref{thm:trace}, which implies, with (\ref{eq:reg_estimate}), that 
\begin{equation}\label{eq:step1-tr} \| \omega_{N,t}^k - \wt{\omega}_{N,t}^k \|_\text{tr} \leq C N \eps \exp ( C \exp ( C |t|))  \left[ 1 + (\eps \sqrt{k})^4 \right] \end{equation}
for a constant $C$ depending on $\sup_N \| W_N \|_{H^2_2}$. 

To estimate the difference $\omega_{N,t}^k - \wt{\omega}_{N,t}^k$ in Hilbert-Schmidt norm, we proceed as we did to derive \eqref{eq:tr_1}. We find
\begin{equation}\label{eq:hs_1}
\begin{split} 
\| \omega_{N,t}^k- \widetilde{\omega}_{N,t}^k\|_{\mathrm{HS}}&\leq  \frac{1}{ \varepsilon} \int_0^t d s \; \| [ \sqrt{1-\varepsilon^2 \Delta} ,\widetilde{\omega}_{N,s}^k ] - A^k_{N,s} \|_{\mathrm{HS}} \\
	&+\frac{1}{ \varepsilon} \int_0^t d s \; \| [ V * (\rho_s^k-\widetilde{\rho}_s^k), \widetilde{\omega}_{N,s}^k ] \|_{\mathrm{HS}} 
	+\frac{1}{\varepsilon} \int_0^t d s \; \|C^k_{N,s} \|_{\mathrm{HS}} 
\end{split} \end{equation}
where $A_{N,s}$ is the operator with the momentum space integral kernel 
\[ \widehat{A}_{N,s} (p;q) =  \frac{\eps^2 (p-q) \cdot (p+q)}{2 \sqrt{1- \frac{\eps^2}{4} (p-q)^2}} \, \widehat{\wt{\omega}}^k_{N,t} (p;q) \]
while $C_{N,s}$ is the operator with the position space kernel 
\[ C_{N,s} (x;y) =  \left[ ( V * \widetilde{\rho}^k_t) (x) - (V * \widetilde{\rho}^k_t)(y) - \nabla (V * \widetilde{\rho}^k_t)\left( \frac{x+y}{2}\right) \cdot (x-y) \right] \widetilde{\omega}^k_{N,t}(x;y)  \]

Proceeding as in \cite[(4.5)-(4.9)]{BPSS}, using (\ref{eq:step1-tr}) (to estimate the difference $\rho^k_s - \wt{\rho}^k_s$) and the propagation of regularity for the solution of the relativistic Vlasov equation established in Prop. \ref{prop:regest} below, we show that the last two summands on the r.h.s. of (\ref{eq:hs_1}) are bounded by 
\begin{equation}
\| [ V * (\rho_s^k-\widetilde{\rho}_s^k), \widetilde{\omega}_{N,s}^k ] \|_{\mathrm{HS}}  \leq C \sqrt{N} \varepsilon^2 \exp( C \exp(C |s|)) \left[ 1+ ( \varepsilon \sqrt{k})^4 \right]  \label{eq:hs_bound1}
\end{equation}
and, respectively, by 
 \begin{equation}
 \| C_{N,s} \|_{\mathrm{HS}} \leq C \sqrt{N} \varepsilon^2 e^{C |s|}. \label{eq:hs_bound2}
 \end{equation}
for a constant $C$ depending on $\sup_N \| W_N \|_{H^2_2}$. 

To control the first term in \eqref{eq:hs_1}, we notice that, in momentum space, the operator $[\sqrt{1-\eps^2 \Delta}, \, \wt{\omega}_{N,s}^k ] - A_{N,s}^k$ has the integral kernel $F(p;q) \widehat{\wt{\omega}}_{N,s}^k (p;q)$ with $F$ defined as in (\ref{eq:Fpq}). Using the representation (\ref{eq:Fpq2}), we obtain, similarly as in (\ref{eq:F-bds}), the bound
\[ |F(p;q)| \leq C \eps^2 |p-q|^2 \]
which implies that 
\[ \begin{split} \left\| \left[ \sqrt{1-\eps^2 \Delta} , \, \wt{\omega}_{N,s}^k \right] - A_{N,s}^k \right\|^2_\text{HS} \leq \; &C \int dp dq \, \eps^4 \, |p-q|^4 \, \left|\widehat{\wt{\omega}}_{N,s}^k (p;q) \right|^2 \\ \leq \; & C \left\| \left[ \eps \nabla , \left[ \eps \nabla, \wt{\omega}_{N,s}^k \right] \right] \right\|^2_\text{HS} \\ \leq \; &C \eps^2 N \| \wt{W}_{N,s}^k \|_{H^2}^2 \leq C \eps^2 N \| W_N^k \|_{H^2} e^{C|s|} \leq C \eps^2 N e^{C |s|} \end{split} \]
using again Prop. \ref{prop:regest} and (\ref{eq:reg_estimate}). Inserting (\ref{eq:hs_bound1}), (\ref{eq:hs_bound2}) and the last estimate in the r.h.s. of (\ref{eq:hs_1}), we obtain (\ref{eq:step1-HS}). 

\bigskip
{\bf Comparison of $\omega_{N,t}^k$ with $\omega_{N,t}$. } Here, we show that 
\begin{equation}\label{eq:step2-HS}
\| \omega_{N,t} - \omega_{N,t}^k\|_{\mathrm{HS}} \leq C \sqrt{N}  \exp(C |t|)  \left( \varepsilon + \frac{1}{\sqrt{k}}\right)
\end{equation}
adapting the strategy of \cite[Section 3]{BPSS} to the relativistic setting.

Let $\mathcal{U}(t;s)$ be the two-parameter group of unitary transformations satisfying 
\begin{equation}\label{eq:cU-def} i\eps \partial_t \, \cU (t;s) = h_H (t)  \, \cU (t;s) 
\end{equation}
with $\cU (s;s) = 1$ for all $s \in \bR$ and with the time-dependent generator $h_H (t)= \sqrt{1-\varepsilon^2\Delta}+ (V * \rho_t)$ where $\rho_t(x) = N^{-1}
\omega_{N,t}(x;x)$. We find 
\[ \begin{split} 
\omega_{N,t}- \omega_{N,t}^k = \; & \mathcal{U}(t;0) \left( \omega_{N}-\mathcal{U}^*(t;0) \omega_{N,t}^k \mathcal{U}(t;0) \right) \mathcal{U}^*(t;0)\\
	= \; & \mathcal{U}(t;0) \left( \omega_{N}-\omega_{N}^k\right)\mathcal{U}^*(t;0) + \frac{1}{i \varepsilon} \int_0^t d s \; \mathcal{U}(t;s) \left[ V * ( \rho_s-\rho_s^k), \omega_{N,s}^k\right] \mathcal{U}^*(t;s) \end{split} 
\] 
Taking the Hilbert-Schmidt norm, we obtain
\begin{equation}
\|\omega_{N,t}- \omega_{N,t}^k \|_{\mathrm{HS}}
	\leq \| \omega_{N} - \omega_{N}^k\|_{\mathrm{HS}} + \frac{1}{N\varepsilon} \int_0^t d s \int d p \; | \widehat{V}(p) | \left\vert \tr \, e^{-i p \cdot x}( \omega_{N,s}-\omega_{N,s}^k) \right\vert \; \| [e^{i p \cdot x}, \omega_{N,s}^k]\|_{\mathrm{HS}} \label{eq:hs_com2_1}
\end{equation}
Combining the expression
\[  \left[ e^{i p \cdot x} , \omega_{N,s}^k \right] = \int_0^1 d \lambda \, e^{i \lambda p \cdot x } [i p \cdot x, \omega_{N,s}^k] e^{i (1-\lambda) p \cdot x} \]
with Prop. \ref{prop:regest}, we conclude that 
\begin{equation}\label{eq:hs_com2_11} \| [e^{i x \cdot p}, \omega_{N,s}^k]\|_{\mathrm{HS}} \leq C \sqrt{N} \varepsilon |p| \, e^{C |s|}.
\end{equation}
To bound the absolute value of $\tr \,  e^{-i p \cdot x}( \omega_{N,s}-\omega_{N,s}^k)$, we use the next lemma, that extends \cite[Lemma 4.1]{BPSS} to the relativistic case. 
\begin{lemma}
\label{lemma:hs_lemma}
Under the assumptions of Theorem \ref{thm:HS1}, and with $\omega_{N,t}^k$ as defined below (\ref{eq:hs_initial}), there exists a constant $C>0$, depending on $\sup_N \| W_N \|_{H_2^2}$ but not on higher Sobolev norms, such that
\begin{equation}\label{eq:lm-HS}
\sup_{p \in \mathbb{R}^3} \frac{1}{1+ |p|} \left\vert \tr \, e^{i p \cdot x} ( \omega_{N,t}- \omega_{N,t}^k) \right\vert \leq C N e^{C |t|} \left( \frac{1}{\sqrt{k}} + \varepsilon \right). 
\end{equation}
\end{lemma}

Inserting (\ref{eq:hs_com2_11}) and (\ref{eq:lm-HS}) in the r.h.s of (\ref{eq:hs_com2_1}), we obtain (\ref{eq:step2-HS}). 

\begin{proof}[Proof of Lemma \ref{lemma:hs_lemma}.] 
Let $\mathcal{U}(t;s)$ be as defined in (\ref{eq:cU-def}). Then
\begin{align*}
i \varepsilon \partial_t \, \mathcal{U}^*(t;0) \, \omega_{N,t}^k \, \mathcal{U}(t;0)= - \mathcal{U}^*(t;0) \left[ V * ( \rho_t -\rho_t^k), \omega_{N,t}^k\right] \mathcal{U}(t;0),
\end{align*}
leads to
\[ \begin{split} 
\tr \, e^{i p \cdot x} ( \omega_{N,t} - \omega_{N,t}^k) = & \; \tr \, 
\mathcal{U}^*(t;0) e^{i p \cdot x} \mathcal{U}(t;0) \left( \omega_{N} - \mathcal{U}^*(t;0)\omega_{N,t}^k\mathcal{U}(t;0)\right) \\
	= & \;  \tr \, \mathcal{U}^*(t;0) e^{i p \cdot x} \mathcal{U}(t;0) \left( \omega_{N} - \omega_{N}^k\right)- \frac{1}{i \varepsilon}\int_0^t d s \;  \tr \, \mathcal{U}^*(t;s) e^{i p \cdot x} \mathcal{U}(t;s) \left[V * ( \rho_s - \rho_s^k) ,\omega_{N,s}^k \right] \\ 
	= & \;  \tr \, \mathcal{U}^*(t;0) e^{i p \cdot x} \mathcal{U}(t;0) \left( \omega_{N} - \omega_{N}^k\right) \\
	&- \frac{1}{i \varepsilon}\int_0^t d s  \int dq \; \widehat{V}(q) \left( \widehat{\rho}_s(q)-\widehat{\rho}_s^k(q)\right) \, \tr \, \mathcal{U}^*(t;s) e^{i p \cdot x} \mathcal{U}(t;s) \left[e^{i q \cdot x} ,\omega_{N,s}^k \right] 
\end{split} \]
We remark that 
\[
 \widehat{\rho}_s(q)-\widehat{\rho}_s^k(q) = \frac{1}{N} \tr \, e^{-i q \cdot x} \left( \omega_{N,s}- \omega_{N,s}^k \right),
\]
Hence, we obtain
\begin{equation}\label{eq:lm-HS-1} \begin{split} 
\Big| \tr &\; e^{i p \cdot x} ( \omega_{N,t} - \omega_{N,t}^k) \Big|  \\
	\leq \; & \left|   \tr \; \mathcal{U}^*(t;0) e^{i p \cdot x} \mathcal{U}(t;0) \left( \omega_{N} - \omega_{N}^k\right) \right| \\
	&+\frac{1}{N\varepsilon}\int_0^t d s  \int dq \, | \widehat{V}(q)| \left|  \tr \, e^{-i q \cdot x} \left( \omega_{N,s}- \omega_{N,s}^k\right) \right| \; \left| \tr \left[\mathcal{U}^*(t;s) e^{i p \cdot x} \mathcal{U}(t;s), e^{i q \cdot x} \right] \omega_{N,t}^k \right| \end{split} \end{equation}
{F}rom Lemma \ref{lemma:smallr} we have
\[ \sup_{\omega, r} \frac{1}{1+r^2} \left| \tr \, \left[ e^{i r \cdot x} , \cU^* (t;s) \, e^{i p_1 \cdot x + \eps p_2 \cdot \nabla}  \cU (t;s) \right] \omega \right|  \leq C N \eps \,  (|p_1| + |p_2|) \, e^{C |t-s|} \]
where the supremum is taken over all fermionic density matrices $0 \leq \omega \leq 1$, with $\tr \, \omega = N$ and over all $r \in \bR^3$. With the assumption (\ref{eq:assthm2}) on the interaction potential, we conclude from (\ref{eq:lm-HS-1}) that 
\begin{equation}\label{eq:gro-HS} \begin{split} 
\sup_{p \in \bR^3} \frac{1}{1+|p|} \left| \tr \, e^{ip \cdot x} (\omega_{N,t} - \omega_{N,t}^k) \right| \leq \; &\sup_{p \in \bR^3} \frac{1}{1+|p|} \Big| \tr \, \cU^* (t;0)  e^{i p \cdot x} \cU (t;0) \, ( \omega_{N} - \omega_{N}^k) \Big| \\ &+ C \int_0^t e^{C|t-s|} \, \sup_{p \in \bR^3} \frac{1}{1+|p|} \left| \tr \, e^{ip \cdot x} (\omega_{N,s} - \omega_{N,s}^k) \right| \end{split} \end{equation}
The first term on the r.h.s. of the last equation can be bounded as in \cite[(4.17)-(4.21)]{BPSS}. We find:
\[ \sup_{p \in \bR^3} \frac{1}{1+ |p|} \left| \tr \, \cU^* (t;0)  e^{i p \cdot x} \cU (t;0) \, ( \omega_{N} - \omega_{N}^k) \right| \leq CN \left( \eps + \frac{1}{\sqrt{k}} \right) \]
Inserting this bound in (\ref{eq:gro-HS}) and applying Gronwall's lemma, we conclude that 
\[ \sup_{p \in \bR^3} \frac{1}{1+|p|} \left| \tr \, e^{ip \cdot x} (\omega_{N,t} - \omega_{N,t}^k) \right| \leq C N \exp (C|t|) \left( \eps + \frac{1}{\sqrt{k}} \right) \]
as claimed. 
\end{proof}

\bigskip

{\bf Comparison of $\widetilde{\omega}_{N,t}^k$ with $\widetilde{\omega}_{N,t}$.} In this paragraph, we show that 
\begin{equation}\label{eq:step3} \| \wt{\omega}_{N,t} - \wt{\omega}_{N,t}^k \|_\text{HS} = \sqrt{N} \, \| \wt{W}_{N,t} - \wt{W}_{N,t}^k \|_2 \leq C e^{C|t|} \sqrt{\frac{N}{k}} \, \end{equation}

To this end, we denote by $\wt{\rho}_t (x)$ and $\wt{\rho}_t^k (x)$ the position space densities associated with the solutions $\wt{W}_{N,t}$ and $\wt{W}^k_{N,t}$ of the Vlasov equation with initial data $W_N$ and $W_N^k$. We consider the solutions $(X_t (x,v), V_t (x,v))$ and $(X^k_t (x,v), V^k_t (x,v))$ of Newton's equations
\begin{align}\label{eq:cases}
\begin{cases}
\dot{X}_t(x,v) = \frac{V_t(x,v)}{\sqrt{1+V_t^2(x,v)}}\\
\dot{V}_t(x,v) = - \nabla ( V * \widetilde{\rho}_t) ( X_t(x,v))
\end{cases} \qquad \text{and} \qquad 
\begin{cases}
\dot{X}^k_t (x,v) = \frac{V^k_t(x,v)}{\sqrt{1+V^k_t(x,v)^2}}\\
\dot{V}^k_t (x,v) = - \nabla ( V * \widetilde{\rho}^k_t) ( X^k_t(x,v))
\end{cases},
\end{align}
with the initial conditions $X_0 (x,v) = X^k_0 (x,v) = x$ and $V_0 (x,v) = V_0^k (x,v) = v$. {F}rom (\ref{eq:cases}), we find
\begin{equation}\label{eq:XtVt} \begin{split} X_t (x,v) &= x + \int_0^t \frac{V_s(x,v)}{\sqrt{1+V_s^2(x,v)}} ds, \qquad 
V_t (x,v) = v - \int_0^t \nabla ( V * \widetilde{\rho}_s) ( X_s (x,v)) ds \end{split} \end{equation}
which easily imply, with the assumption (\ref{eq:assthm2}) on the interaction potential, that 
\[ \begin{split} |\nabla_x X_t (x,v)| &\leq C + C \int_0^t |\nabla_x V_s (x,v)| ds, \qquad 
|\nabla_x V_t (x,v)| \leq C \int_0^t | \nabla_x X_s (x,v)| ds , \\  |\nabla_v X_t (x,v)| &\leq C \int_0^t |\nabla_v V_s (x,v)| ds, \hspace{1.3cm} 
|\nabla_v V_t (x,v)| \leq C + C \int_0^t |\nabla_v X_s (x,v)| ds , \end{split} \]
and therefore, by Gronwall, that there exists $C > 0$ such that 
\begin{equation}\label{eq:nablaxvXV} |\nabla_x X_t (x,v)| + | \nabla_v X_t (x,v)| + |\nabla_x V_t (x,v)| + |\nabla_v V_t (x,v)| \leq C e^{C|t|} \end{equation}
Similarly, we find 
\begin{equation}\label{eq:nablakxvXV} |\nabla_x X^k_t (x,v)| + | \nabla_v X^k_t (x,v)| + |\nabla_x V^k_t (x,v)| + |\nabla_v V^k_t (x,v)| \leq C e^{C|t|} \end{equation}
uniformly in $k>0$. {F}rom (\ref{eq:XtVt}) and from the similar expression for $X_t^k, V_t^k$, it is also simple to check that 
\[ \begin{split} |X_t (x,v) - X^k_t (x,v)| &\leq C \int_0^t  |V_s (x,v) - V^k_s (x,v)| ds , \\ |V_t (x,v) - V_t^k (x,v)| &\leq C \int_0^t \| \wt{\rho}_s - \wt{\rho}_s^k \|_1 \, ds + C \int_0^t |X_s (x,v) - X_s^k (x,v)| ds \end{split} \] 
Since, by definition of $\wt{\rho}_t$ and $\wt{\rho}_t^k$, $\| \wt{\rho}_s - \wt{\rho}_s^k \|_1 \leq \| \wt{W}_{N,s} - \wt{W}_{N,s}^k \|_1$, we conclude by Gronwall's lemma that
\begin{equation}\label{eq:XminX} |X_t (x,v) - X^k_t (x,v)| + |V_t (x,v) - V_t^k (x,v)| \leq C e^{C|t|} \int_0^t \| \wt{W}_{N,s} - \wt{W}_{N,s}^k \|_1 ds  \, . \end{equation}
We also need estimates on the differences among derivatives of $X_t, V_t$ and of $X_t^k, V_t^k$. Starting again with the expressions (\ref{eq:XtVt}), using the bound
\[ \| \nabla^3 (V * \wt{\rho}_s) \|_\infty \leq \| \nabla^2 V \|_\infty \| \nabla \wt{W}_{N,s} \|_1 \leq C \| \wt{W}_{N,s} \|_{H^1} \leq C e^{C|s|} \| W_N \|_{H^1} \]
(where the last estimate follows from Prop. \ref{prop:regest}) and proceeding similarly as in \cite[(4.30)-(4.32)]{BPSS}, we conclude that 
\begin{equation}\label{eq:nablaxXV} \begin{split} \left|\nabla_x X_t (x,v) - \nabla_x X_t^k (x,v)\right| &+ \left| \nabla_x V_t (x,v) - \nabla_x V_t^k (x,v)\right| \\ &\leq C e^{C|t|} \int_0^t ds \, \| \wt{\rho}_s - \wt{\rho}^k_s \|_1 +  C e^{C|t|} \int_0^t ds \, \int_0^s d\tau \, \| \wt{\rho}_\tau - \wt{\rho}^k_\tau \|_1 \end{split} \end{equation}
and that 
 \begin{equation}\label{eq:nablavXV} \begin{split} \left|\nabla_v X_t (x,v) - \nabla_v X_t^k (x,v)\right| &+ \left| \nabla_v V_t (x,v) - \nabla_v V_t^k (x,v)\right| \\ &\leq C e^{C|t|} \int_0^t ds \, \| \wt{\rho}_s - \wt{\rho}^k_s \|_1 +  C e^{C|t|} \int_0^t ds \, \int_0^s d\tau \, \| \wt{\rho}_\tau - \wt{\rho}^k_\tau \|_1 \end{split} \end{equation}

With the bounds (\ref{eq:nablaxvXV}), (\ref{eq:nablakxvXV}), (\ref{eq:XminX}),  (\ref{eq:nablaxXV}), (\ref{eq:nablavXV}), we can apply the same continuity argument used in \cite{BPSS}[(4.32)-(4.41)]. The bound 
\[ \| W_N - W_N^k \|_1 \leq \| W_N - W_N^k \|_{H^0_2} \leq C / \sqrt{k} \] implies, first of all, that
\[ \| \wt{W}_{N,t} - \wt{W}_{N,t}^k \|_1 \leq C \frac{e^{C|t|}}{\sqrt{k}} \]
Using this estimate and the bound $\| W_N - W_N^k \|_2 \leq C / \sqrt{k}$, we also obtain  
\[ \| \wt{W}_{N,t} - \wt{W}_{N,t}^k \|_2 \leq C \frac{e^{C|t|}}{\sqrt{k}}  \]

\bigskip

{\bf Proof of Theorem \ref{thm:HS1}.} Combining (\ref{eq:step1-HS}), (\ref{eq:step2-HS}) and (\ref{eq:step3}), we obtain 
\[ \begin{split} \| \omega_{N,t} - \wt{\omega}_{N,t} \|_\text{HS} &\leq \| \omega_{N,t} - \omega^k_{N,t} \|_\text{HS} + \| \omega_{N,t}^k - \wt{\omega}_{N,t}^k \|_\text{HS} + \| \wt{\omega}_{N,t}^k - \wt{\omega}_{N,t} \|_\text{HS} \\ & \leq C N^{1/2} \left[ \eps + \frac{1}{\sqrt{k}} \right] \exp (C \exp (C|t|))  \left[ 1 + (\eps \sqrt{k})^4 \right] \end{split}  \]
Choosing in particular $k = \eps^{-2}$ we conclude that
\[   \| \omega_{N,t} - \wt{\omega}_{N,t} \|_\text{HS} \leq C \sqrt{N} \, \eps  \exp (C \exp (C|t|)) \]
as claimed.

\bigskip

{\bf Proof of Theorem \ref{thm:HS2}.} Recall that $W_{N,t}$ denotes the Wigner transform of the solution $\omega_{N,t}$ of Hartree's equation and that $W_N$ denotes the Wigner transform of the initial data $\omega_{N,t=0} = \omega_N$. Furthermore, $W_t$ denotes the solution of the Vlasov equation with initial data $W_0$. 

We define $\wt{W}_{N,t}$ the solution of the Vlasov equation with initial data $W_N$. Then we have
\begin{equation}\label{eq:triangle} 
\| W_{N,t} - W_t \|_2 \leq \| W_{N,t} - \wt{W}_{N,t} \|_2 + 
\| \wt{W}_{N,t} - W_t \|_2 
\end{equation}
Theorem \ref{thm:HS1} implies that 
\[ \| W_{N,t} - \wt{W}_{N,t} \|_2 \leq C \eps \exp ( C \exp (C|t|)) \]
To bound the second term on the r.h.s. of (\ref{eq:triangle}), we have to compare two solutions $\wt{W}_{N,t}$ and $W_t$ of the Vlasov equation associated to the two initial data $W_N$ and $W_0$ in terms of the norms 
$\| W_N - W_0 \|_1 = \kappa_{N,1}$ and $\| W_N - W_0 \|_2 = \kappa_{N,2}$.
This is exactly what we did to prove (\ref{eq:step3}). Following the same strategy (which, as explain above, goes back to \cite{BPSS}), we obtain that
\[  \| \wt{W}_{N,t} - W_t \|_2  \leq C (\kappa_{N,1} + \kappa_{N,2} ) \exp (C |t|) \]
Hence, (\ref{eq:triangle}) implies that 
\[ \|W_{N,t} - W_t \|_2 \leq   C \eps \exp ( C \exp (C|t|)) + C (\kappa_{N,1} + \kappa_{N,2} ) \exp (C |t|)  \]
and concludes the proof of Theorem \ref{thm:HS2}. 

\section{Convergence for expectation of semiclassical observables}

In this section, we prove Theorem \ref{thm:semi} and Theorem \ref{thm:semi2}. 
To show Theorem \ref{thm:semi}, we first make the additional assumption $\| W_N \|_{H^4_2} < C$, uniformly in $N$. 

{F}rom \eqref{eq:tr0} we find 
\begin{equation}\label{eq:semi_0} 
\begin{split} \tr \; e^{i p \cdot x + q \cdot \varepsilon \nabla} ( \omega_{N,t}- \widetilde{\omega}_{N,t}) = & \; \frac{1}{i \varepsilon} \int_0^t ds \; \tr \; e^{i p \cdot x + q \cdot \varepsilon \nabla} \mathcal{U}(t;s) \left( \left[ \sqrt{1-\varepsilon^2\Delta}, \widetilde{\omega}_{N,s} \right] - A_{N,s} \right) \mathcal{U}^*(t;s) \\
	&+ \frac{1}{i \varepsilon} \int_0^t ds \; \tr \; e^{i p \cdot x + q \cdot \varepsilon \nabla} \mathcal{U} (t;s) [V * ( \rho_s-\widetilde{\rho}_s), \widetilde{\omega}_{N,s} ] \mathcal{U}^*(t;s) \\
	&+ \frac{1}{i \varepsilon} \int_0^t ds \; \tr e^{i p \cdot x + q \cdot \varepsilon \nabla} \mathcal{U} (t;s) \, C_{N,s} \mathcal{U}^*(t;s),
\end{split} \end{equation}
where $\cU (t;s)$ is the unitary dynamics defined in (\ref{eq:Uts}), $\rho_s (x) = N^{-1} \omega_{N,s} (x;x)$, $\wt\rho_s (x) = N^{-1} \wt{\omega}_{N,s} (x;x)$, the operator $A_{N,s}$ has the momentum space kernel 
\[ \widehat{A}_{N,s} (p;q) = \frac{\varepsilon^2 (p-q) \cdot (p+q)}{2 \sqrt{1- \frac{\eps^2}{4} (p-q)^2}} \; \widehat{\wt{\omega}}_{N,t} (p;q) \]
and the operator $C_{N,s}$ has the position space kernel
\[ C_{N,s} (x;y) = \left[ (V * \wt{\rho}_t) (x) -(V * \wt{\rho}_t) (y) - \nabla (V* \wt{\rho}_t) \left( \frac{x+y}{2} \right) \cdot (x-y) \right] \wt{\omega}_{N,t} (x;y) \, . \]
We consider, first of all, the contribution arising from the second term on the r.h.s. of (\ref{eq:semi_0}). We write 
\[ \begin{split} 
\tr \; e^{i p \cdot x + q \cdot \varepsilon \nabla} & \mathcal{U}(t;s) [V * (\rho_s-\widetilde{\rho}_s ), \widetilde{\omega}_{N,s}] \, \mathcal{U}^*(t;s) 
\\
=& \int d z \; ( \rho_s(z)-\widetilde{\rho}_s(z)) \; \tr \; e^{i p \cdot x + q \cdot \varepsilon \nabla} \mathcal{U}(t;s) [V(x-z), \widetilde{\omega}_{N,s}]\, \mathcal{U}^*(t;s) \\
=& \frac{1}{N} \int d k \; \widehat{V}(k) \; \tr \; e^{-i k \cdot z} ( \omega_{N,s}-\widetilde{\omega}_{N,s}) \tr \, e^{i p \cdot x + q \cdot \varepsilon \nabla} \mathcal{U}(t;s) [e^{i x \cdot k}, \widetilde{\omega}_{N,s}]\, \mathcal{U}^*(t;s)\\
\end{split}\]
Hence, we obtain (using the assumption (\ref{eq:assthm4}) on the interaction $V$) 
\begin{equation}\label{eq:semi-sec} \begin{split} \Big| \tr \; e^{i p \cdot x + q \cdot \varepsilon \nabla} & \mathcal{U}(t;s) [V * (\rho_s-\widetilde{\rho}_s ), \widetilde{\omega}_{N,s}] \, \mathcal{U}^*(t;s)  \Big| \\ \leq \; & \frac{1}{N} \int dk |\widehat{V} (k)| \Big| \tr\, e^{ik\cdot x} (\omega_{N,s} - \wt{\omega}_{N,s}) \Big| \, \left| \tr \, e^{ip \cdot x + q \cdot \eps \nabla} \, \cU (t;s) [ e^{ik \cdot x} , \wt{\omega}_{N,s} ] \, \cU^* (t;s) \right| \\ \leq \; & \frac{C \, \tr \,  |\, \wt{\omega}_{N,s} \, |}{N} \sup_{k \in \bR^3} \frac{1}{1+k^2} \Big| \tr\, e^{ik \cdot x} (\omega_{N,s} - \wt{\omega}_{N,s}) \Big| \sup_{\omega,k} \frac{1}{1+k^2} \, \left| \tr \, [ e^{ik \cdot x} , \cU^* (t;s) \, e^{ip \cdot x + q \cdot \eps \nabla} \, \cU (t;s)] \omega \right| \end{split} \end{equation}
where the last supremum is taken over $k \in \bR^3$ and over all operators trace class operators $\omega$ on $L^2 (\bR^3)$ with $\tr \, |\omega| \leq 1$. {F}rom Lemma \ref{lemma:smallr}, we find 
\begin{equation}\label{eq:semi-secf} 
\begin{split} 
\Big| \tr \; e^{i p \cdot x + q \cdot \varepsilon \nabla} & \mathcal{U}(t;s) [V * (\rho_s-\widetilde{\rho}_s ), \widetilde{\omega}_{N,s}] \, \mathcal{U}^*(t;s)  \Big| \\ &\leq \frac{C \, \tr \, |\wt{\omega}_{N,s}|}{N} \eps (|p|+|q|) e^{C|t-s|} \, \sup_{k \in \bR^3} \frac{1}{(1+|k|)^2} \left| \tr \, e^{ik \cdot x} (\omega_{N,s} - \wt{\omega}_{N,s}) \right| \end{split} \end{equation}

Next, let us focus on the last term on the r.h.s. of (\ref{eq:semi_0}). Following the same strategy as in \cite{BPSS}[(5.5)-(5.6)], we have
\[  \Big| \, \tr \, e^{ip \cdot x+ q \cdot \eps \nabla} \, \cU (t;s) C_s \, \cU^* (t;s) \Big| \leq C \, \tr \, |\wt{\omega}_{N,s}| |\sup_{\omega, i,j} \left| \tr \, \left[ x_i , \left[ x_j, \cU^* (t;s) \, e^{ip \cdot x + q \cdot \eps \nabla} \cU (t;s) \right] \right] \omega \right| \]
With Lemma \ref{lemma:double}, we conclude that 
\begin{equation}\label{eq:semi-thirf} \Big| \, \tr \, e^{ip \cdot x+ q \cdot \eps \nabla} \, \cU (t;s) C_s \, \cU^* (t;s) \Big| \leq C \eps^2 \, (|p| + |q|)^2 e^{C |t-s|} \, \tr \, |\wt{\omega}_{N,s}| \end{equation}

Finally, we consider the first term on the r.h.s. of (\ref{eq:semi_0}). 
Using (\ref{eq:Fpq2}), we rewrite the momentum space kernel of the operator $[ \sqrt{1- \eps^2 \Delta} , \wt{\omega}_{N,s}] - A_{N,s}$ as follows 
\[  \begin{split} 
\Big( [ \sqrt{1- \eps^2 \Delta} , & \; \wt{\omega}_{N,s}] - A_{N,s} \Big) (p;q) 
\\ = \; & \left( \sqrt{1+\eps^2 p^2} - \sqrt{1+ \eps^2 q^2} - \frac{\eps^2 (p-q) \cdot (p+q)}{\sqrt{1+ \frac{\eps^2}{4} (p-q)^2}} \right) \wt{\omega}_{N,s} (p;q) \\ = \; & \int_0^1 d\kappa \int_0^\kappa dr \left\{ \eps^2 (p-q)^2 \left[ \frac{1}{\left( 1 + \eps^2 \left( \frac{p+q}{2} + r \frac{p-q}{2} \right)^2 \right)^{1/2}} - \frac{1}{\left( 1 + \eps^2 \left( \frac{p+q}{2} - r \frac{p-q}{2} \right)^2 \right)^{1/2}} \right] \right. \\ 
\; & \left. + \eps^4 \left[ \frac{\left( (p-q) \cdot \left( \frac{p+q}{2} + r \frac{p-q}{2} \right) \right)^2}{\left( 1 + \eps^2 \left( \frac{p+q}{2} + r \frac{p-q}{2}\right)^2 \right)^{3/2}} - \frac{\left( (p-q) \cdot \left( \frac{p+q}{2} - r \frac{p-q}{2}\right)\right)^2}{\left( 1 + \eps^2 \left( \frac{p+q}{2} - r \frac{p-q}{2}\right)^2 \right)^{3/2}} \right] \right\} \, \wt{\omega}_{N,s} (p;q) \end{split} \]

We write
\begin{equation}\label{eq:fij-rep}  \begin{split} 
\Big( [ \sqrt{1- \eps^2 \Delta} ,  &\, \wt{\omega}_{N,s}] - A_{N,s} \Big) (p;q) \\ = \; & \int_0^1 d\kappa \int_0^\kappa dr \sum_{i,j = 1}^3 \eps (p_i - q_i) \eps (p_j - q_j) \\ & \hspace{2cm} \times \left[ f_{ij} \left( \eps \left( \frac{p+q}{2} + r \frac{p-q}{2} \right) \right) - f_{ij} \left( \eps \left( \frac{p+q}{2} - r \frac{p-q}{2} \right) \right) \right] \wt{\omega}_{N,s} (p;q) 
\end{split} \end{equation}
with the definition 
\[ f_{ij} (\xi) = \frac{1}{\sqrt{1+\xi^2}} \left[ \delta_{ij} + \frac{\xi_i \xi_j}{1+\xi^2} \right]  \]
for all $i,j = 1,2,3$ and $\xi \in \bR^3$. We consider the (distributional) Fourier transform
\[ \widehat{f}_{ij} (k) = \frac{1}{(2\pi)^3} \int d\xi \, f_{ij} (\xi) e^{-i \xi \cdot k} \]
We notice that 
\[ \begin{split} |k|^4 \widehat{f}_{ij} (k) &= \frac{1}{(2\pi)^3} \int d\xi f_{ij} (\xi) \Delta^2_\xi e^{i k \cdot \xi} = \frac{1}{(2\pi)^3} \int d\xi \Delta^2_\xi f_{ij} (\xi) e^{i k \cdot \xi} \end{split} \]
Since $|e^{ik \cdot \xi}| \leq 1$, we immediately find (using the fact that $\Delta^2 f_{ij}$ is integrable at infinity) that $|k|^4 |\widehat{f}_{ij} (k)| \leq C$. On the other hand, since, for example, \[ |e^{ik \cdot x} - 1- i k \cdot x| \leq C |k|^{3/2} |x|^{3/2}  \]
we also have $|k|^4 |\widehat{f}_{ij} (k)| \leq C |k|^{3/2} $ (because $|\xi|^{3/2} |\Delta^2 f_{ij} (\xi)|$ is also integrable at infinity). Combining these two bounds, we conclude that
\begin{equation}\label{eq:fij-bd} |\widehat{f}_{ij} (k)| \leq \frac{C}{|k|^{5/2} (1 + |k|^{3/2})} \end{equation}
for all $i,j = 1,2,3$. 
  
{F}rom (\ref{eq:fij-rep}), we obtain 
\[ \begin{split} 
\Big( [ \sqrt{1- \eps^2 \Delta} , & \; \wt{\omega}_{N,s}] - A_{N,s} \Big) (p;q) \\ = \; & \sum_{i,j =1}^3 \int_0^1 d\kappa \int_0^\kappa dr \int dk \, \widehat{f}_{ij} (k) \eps (p_i - q_i) \eps (p_j - q_j) e^{i \eps k \cdot \left( \frac{p+q}{2} + r \frac{p-q}{2} \right)} \wt{\omega}_{N,s} (p;q) \\ &- 
\sum_{i,j =1}^3 \int_0^1 d\kappa \int_0^\kappa dr \int dk \, \widehat{f}_{ij} (k) \eps (p_i - q_i) \eps (p_j - q_j) e^{i \eps k \cdot \left( \frac{p+q}{2} - r \frac{p-q}{2} \right)} \wt{\omega}_{N,s} (p;q) \\
=\; & \sum_{i,j =1}^3 \int_0^1 d\kappa \int_0^\kappa dr \int dk \, \widehat{f}_{ij} (k) \, \left[ \eps \nabla_i , \left[ \eps \nabla_j , e^{\frac{1+r}{2} k \cdot \eps \nabla} \wt{\omega}_{N,s} e^{-\frac{1-r}{2} k \cdot \eps \nabla} \right] \right] (p;q) \\ & - \sum_{i,j =1}^3 \int_0^1 d\kappa \int_0^\kappa dr \int dk \, \widehat{f}_{ij} (k) \, \left[ \eps \nabla_i , \left[ \eps \nabla_j , e^{\frac{1-r}{2} k \cdot \eps \nabla} \wt{\omega}_{N,s} e^{-\frac{1+r}{2} k \cdot \eps \nabla} \right] \right] (p;q)
\end{split} \]
Using (\ref{eq:fij-bd}) to integrate over $k$, we conclude that 
\[ \begin{split} \Big| \tr \; e^{ip \cdot x + q \cdot \eps \nabla} \, & \cU (t;s) \left( \left[ \sqrt{1-\eps^2 \Delta} , \wt{\omega}_{N,s} \right] - A_{N,s} \right) \, \cU^* (t;s) \Big| \\ \leq \; & C \; \tr \, | \, \wt{\omega}_{N,s} | \, \sup_{\omega, i,j} \left| \tr \; e^{ip \cdot x + q \cdot \eps \nabla} \, \cU (t;s) \left[ \eps \nabla_i , \left[ \eps \nabla_j , \omega \right] \right] \cU^* (t;s) \right| 
\\ = \; & C \; \tr \, | \, \wt{\omega}_{N,s} | \, \sup_{\omega, i,j} \left| \tr \;  \left[ \eps \nabla_i , \left[ \eps \nabla_j , \cU^* (t;s) \, e^{ip \cdot x + q \cdot \eps \nabla} \, \cU (t;s) \right] \right] \omega \right| 
\end{split} \]
where the supremum is taken over all $i,j = 1,2,3$ and over all trace class operators $\omega$ on $L^2 (\bR^3)$, with $\tr \, |\omega| \leq 1$. Applying Lemma \ref{lemma:double}, we obtain 
\begin{equation}\label{eq:semi-firf} \begin{split} \Big| \tr \; e^{ip \cdot x + q \cdot \eps \nabla} \, & \cU (t;s) \left( \left[ \sqrt{1-\eps^2 \Delta} , \wt{\omega}_{N,s} \right] - A_{N,s} \right) \, \cU^* (t;s) \Big| \leq C \eps^2 (|p|+|q|)^2 e^{C|t-s|}  \, \tr \, |\wt{\omega}_{N,s}| \end{split} \end{equation}

Bounding the r.h.s. of (\ref{eq:semi_0}) with (\ref{eq:semi-secf}), (\ref{eq:semi-thirf}) and (\ref{eq:semi-firf}), we obtain 
\begin{equation}\label{eq:bd-allto} \begin{split} 
\Big| \tr \; e^{ip \cdot x + i q \cdot \eps \nabla} \, \left( \omega_{N,t} - \wt{\omega}_{N,t} \right) \Big| \leq \; &C  \eps (|p| + |q|)^2 \, \int_0^t \, e^{C|t-s|} \, \tr \, |\wt{\omega}_{N,s}| \, ds \\ &+ \frac{C (|p|+|q|)}{N} \int_0^t e^{C|t-s|} \, \tr\, |\wt{\omega}_{N,s}| \, \sup_k  \frac{1}{(1+|k|)^2} \Big| \tr \, e^{ik \cdot x} (\omega_{N,s} - \wt{\omega}_{N,s}) \Big| \end{split} \end{equation}

Proceeding as in \cite{BPSS}[(5.10)-(5.11)], we find
\[ \tr \, |\wt{\omega}_{N,s}| \leq C N e^{C|s|} \sum_{j=0}^4 \eps^j \| W_N \|_{H^j_2} \]
Inserting in (\ref{eq:bd-allto}) and applying Gronwall's lemma, we arrive at
\[ \begin{split} \sup_{p,q \in \bR^3} \frac{1}{(1+|p|+|q|)^2} &\Big| \tr\; e^{ip \cdot x + i q \cdot \eps \nabla} \, \left( \omega_{N,t} - \wt{\omega}_{N,t} \right) \Big| \\ & \hspace{2cm} \leq C N \eps \left[ \sum_{j=0}^4 \eps^j \| W_N \|_{H^j_2} \right]  \exp \left( C \left[ \sum_{j=0}^4 \eps^j \| W_N \|_{H^j_2} \right] \exp (C |t|) \right) \end{split} \]
which proves Theorem \ref{thm:semi}, under the additional assumption that $\| W_N \|_{H^4_2} < \infty$, uniformly in $N$. To relax this condition, we proceed analogously as in \cite{BPSS} (starting from Eq. (5.12)), with a simple approximation argument (it is in this approximation argument that the assumption (\ref{eq:comms}) plays an important role). 

\bigskip

Also the proof of Theorem \ref{thm:semi2} follows the line of the proof of Theorem 2.5 in \cite{BPSS}. Recall that $\widehat{W}_{N,t}$ is the Fourier transform of the Wigner transform $W_{N,t}$ of the solution $\omega_{N,t}$ of Hartree's equation with initial data $\omega_N$. $W_t$, on the other hand, denotes the solution of the relativistic Vlasov equation with initial data $W_0$. We introduce the notation $\wt{W}_{N,t}$ for the solution of Vlasov equation with initial data $W_N$, given by the Wigner transform of $\omega_{N}$. It follows from Theorem \ref{thm:semi} that 
\[\begin{split} \Big| \widehat{W}_{N,t} (p;q) - \widehat{W}_t (p;q) \Big| \leq   \; & \Big| \widehat{W}_{N,t} (p;q) - \widehat{\wt{W}}_{N,t} (p;q) \Big| + \Big| \widehat{\wt{W}}_{N,t} (p;q) - \widehat{W}_t (p;q) \Big| \\  \leq \; &C \eps (1 + |p| + |q|)^2 \exp (C \exp (C|t|))  +  \Big| \widehat{\wt{W}}_{N,t} (p;q) - \widehat{W}_t (p;q) \Big| \end{split} \]
Following the arguments in \cite{BPSS}, in particular, the proof of the stability of Vlasov equation between (4.23) and (4.39), we find that 
\[ \| \wt{W}_{N,t} - W_t \|_1 \leq C \exp (C \exp (C |t|))  \| W_N - W_0 \|_1 \]
for a constant $C > 0$ depending only on $\| W_0 \|_{W^{1,1}}$. This implies that 
\[ \left| \widehat{W}_{N,t} (p;q) - \widehat{W}_t (p;q) \right| \leq C (1 + |p| + |q|)^2 (\eps+ \kappa_N) \exp (C \exp (C|t|)) \]

\section{Auxiliary results} 
We prove the following auxiliary Lemmata as in \cite [Section 6]{BPSS}, whereby the relativistic dynamics affects some changes. 
The first Lemma is an adaption of \cite[Lemma 4.2]{BPSS} to the relativistic dynamics.  
 \begin{lemma}
 \label{lemma:smallr}
Assume that $V \in L^1\left( \mathbb{R}^3\right)$ satisfies
\begin{align}
\int d p \; ( 1+|p|^3)  | \widehat{V}(p) | < \infty. \label{ass:lemma_smallr}
\end{align}
Let $\mathcal{U}(t;s)$ be the unitary evolution generated by the relativistic Hartree Hamiltonian $h_H(t)= \sqrt{1-\varepsilon^2 \Delta} + \left( V * \rho_t \right)$. There exists a constant $C>0$ such that 
\[ \begin{split} 
\sup_{\omega, r \in \bR^3} \frac{1}{1 + r^2} \left| \, \tr \left[ e^{i r \cdot x }, \mathcal{U}^*(t;s) e^{i x \cdot p + \varepsilon \nabla \cdot q} \mathcal{U}(t;s)\right] \omega \right| \leq \; & C \varepsilon \left( |p| + |q| \right) e^{C |t-s|}\\
\sup_{\omega} \left| \, \tr \left[ \eps \nabla , \mathcal{U}^*(t;s) e^{i x \cdot p + \varepsilon \nabla \cdot q} \mathcal{U}(t;s)\right] \omega \right|  \leq \; & C \varepsilon \left( |p| + |q| \right) e^{C |t-s|},
\end{split} \]
where the suprema are taken over all trace class operators $\omega$ with $\tr|\omega| \leq 1$.
\end{lemma}

\begin{proof} We will apply Gronwall's lemma. First of all, we define a unitary dynamics $\wt{\cU} (t;s)$ (a two-parameter group of unitary transformations) through
\[ i \eps \partial_t \wt{\cU} (t;s) = e^{i r \cdot x}   h_H (t) e^{-i r \cdot x} \wt{\cU} (t;s) = \left[ h_H (t) + A \right] \wt{\cU} (t;s) \]
with the operator 
\[ A = \sqrt{1 + \eps^2 (i \nabla - r)^2} - \sqrt{1 - \eps^2 \Delta} \]
We observe that 
\begin{equation}\label{eq:lm1-0} \begin{split}  \sup_\omega \Big| \tr \, \left[ e^{ir \cdot x}  , \cU^* (t;s)  e^{i x \cdot p + \varepsilon \nabla \cdot q} \cU (t;s) \right] \omega \Big|  = \; & \sup_\omega \Big| \tr \, \left[ e^{ir \cdot x}  , \cU^* (t;s)  e^{i x \cdot p + \varepsilon \nabla \cdot q} \cU (t;s) \right] \,  \cU (s;0) \, \omega \, \wt{\cU}^* (s;0) \, \Big| \\
= \; & \sup_\omega \Big| \tr \; \wt{\cU}^* (s;0) \left[ e^{ir \cdot x}  , \cU^* (t;s)  e^{i x \cdot p + \varepsilon \nabla \cdot q} \cU (t;s) \right] \, \cU (s;0) \, \omega \, \Big| \end{split} \end{equation}
where suprema are taken over all trace class operators $\omega$ with $\tr |\omega| \leq 1$. 

For fixed $\omega$ (with $\tr \, |\omega| \leq 1$) and fixed $t \in \bR$, we compute
\[ \begin{split} 
i \eps \partial_s \, \tr \; \wt{\cU}^* (s;0)   &\left[ e^{ir \cdot x}, \cU^* (t;s)  e^{i x \cdot p + \varepsilon \nabla \cdot q} \cU (t;s) \right] \cU (s;0) \, \omega \\ = \; &-\tr \; \wt{\cU}^* (s;0) \, \left[ h_H (s) ,  \left[ e^{ir \cdot x}, \cU^* (t;s)  e^{i x \cdot p + \varepsilon \nabla \cdot q} \cU (t;s) \right] \right] \cU (s;0) \, \omega \\
&- \tr \; \wt{\cU}^* (s;0) \, A \left[ e^{ir \cdot x}, \cU^* (t;s)  e^{i x \cdot p + \varepsilon \nabla \cdot q} \cU (t;s) \right] \cU (s;0) \, \omega \\
&+ \tr \, \wt{\cU}^* (s;0)  \left[ e^{i r\cdot x} , \left[ h_N (s), \cU^* (t;s) e^{i x \cdot p + \varepsilon \nabla \cdot q}  \cU (t;s) \right] \right] \cU (s;0) \, \omega \\
= \; & \tr \; \wt{\cU}^* (s;0) \, \left[ \cU^* (t;s) e^{i x \cdot p + \varepsilon \nabla \cdot q} \cU( t;s) , A \right] e^{ir \cdot x} \, \cU (s;0) \, \omega \end{split} \]
In the last step, we used Jacobi's identity. Integrating from $s$ to $t$, we find
\[ \begin{split} 
\tr \; \wt{\cU}^* (s;0) \, &\left[ e^{i r \cdot x} , \cU^* (t;s) e^{i x \cdot p + \varepsilon \nabla \cdot q} \cU (t;s) \right]  \cU (s;0) \omega \\ = \; & \tr \; \wt{\cU}^* (t;0) \left[ e^{i r\cdot x} , e^{i x \cdot p + \varepsilon \nabla \cdot q}  \right] \cU (t;0) \, \omega \\ &+ \frac{1}{i\eps} \int_s^t d\tau \, \tr \; \wt{\cU}^* (\tau, 0) \left[ \cU^* (t;\tau) e^{i x \cdot p + \varepsilon \nabla \cdot q} \cU (t;\tau) , A \right] e^{i r\cdot x} \, \cU (\tau; 0) \, \omega \end{split} \] 
Since
\[ \left[ e^{ir \cdot x}, e^{i x \cdot p + \varepsilon \nabla \cdot q} \right] = \left( e^{-i \eps r \cdot q /2} - e^{i\eps r \cdot q/2} \right) e^{i x \cdot (p+r) + \eps \nabla \cdot q} \]
we obtain the bound
\begin{equation}\label{eq:lm1-A} \begin{split} \Big| \tr \; \wt{\cU}^* (s;0) &\left[ e^{ir \cdot x}, \cU^* (t;s) e^{i x \cdot p + \varepsilon \nabla \cdot q}   \cU (t;s) \right] \cU (s;0) \, \omega \Big| \\ \leq \; &C \eps |r| |q| + \frac{1}{\eps} \int_s^t d\tau \, \left| \tr \; \wt{\cU}^* (\tau,0) \left[ \cU^* (t;\tau) e^{i x \cdot p + \varepsilon \nabla \cdot q} \cU (t;\tau) , A \right] e^{i r\cdot x} \, \cU (\tau; 0) \, \omega \right|  \end{split} \end{equation}
valid for any trace class operator $\omega$ with $\tr \, |\omega| \leq 1$. To estimate the second term on the r.h.s. of (\ref{eq:lm1-A}), we rewrite the commutator
\begin{equation}\label{eq:recom} \begin{split} 
\big[ \cU^* (t;s) &e^{i x \cdot p + \varepsilon \nabla \cdot q}   \cU (t;s) , A \big] \\ = \; &C \int_0^\infty d\kappa \, \sqrt{\kappa} \, \frac{1}{\kappa+1 + \eps^2 (i\nabla - r)^2} \left[ \cU^* (t;s) e^{i x \cdot p + \varepsilon \nabla \cdot q}   \cU (t;s) , \eps^2 (i\nabla - r)^2 \right] \frac{1}{\kappa + 1 + \eps^2 (i\nabla - r)^2} \\ &- C \int_0^\infty d\kappa \, \sqrt{\kappa} \, \frac{1}{\kappa+1 - \eps^2 \Delta} \left[ \cU^* (t;s) e^{i x \cdot p + \varepsilon \nabla \cdot q}   \cU (t;s) , \left( - \eps^2 \Delta \right) \right] \frac{1}{\kappa + 1 - \eps^2 \Delta} \\
 = \; &C \int_0^\infty d\kappa \, \sqrt{\kappa} \, \frac{\eps (i\nabla -r)}{\kappa+1 + \eps^2 (i\nabla - r)^2} \left[ \cU^* (t;s) e^{i x \cdot p + \varepsilon \nabla \cdot q}   \cU (t;s) , \eps \nabla \right] \frac{1}{\kappa + 1 + \eps^2 (i\nabla - r)^2} \\  &+C \int_0^\infty d\kappa \, \sqrt{\kappa} \, \frac{1}{\kappa+1 + \eps^2 (i\nabla - r)^2} \left[ \cU^* (t;s) e^{i x \cdot p + \varepsilon \nabla \cdot q}   \cU (t;s) , \eps \nabla \right] \frac{\eps (i\nabla -r)}{\kappa + 1 + \eps^2 (i\nabla - r)^2} \\ 
&- C \int_0^\infty d\kappa \, \sqrt{\kappa} \, \frac{i \eps \nabla}{\kappa+1 - \eps^2 \Delta} \left[ \cU^* (t;s) e^{i x \cdot p + \varepsilon \nabla \cdot q}   \cU (t;s) , \eps \nabla \right] \frac{1}{\kappa + 1 - \eps^2 \Delta} \\ 
&- C \int_0^\infty d\kappa \, \sqrt{\kappa} \, \frac{1}{\kappa+1 - \eps^2 \Delta} \left[ \cU^* (t;s) e^{i x \cdot p + \varepsilon \nabla \cdot q}   \cU (t;s) , \eps \nabla \right] \frac{i\eps \nabla}{\kappa + 1 - \eps^2 \Delta} \end{split} \end{equation}
for an appropriate constant $C \in \mathbb{C}$. Recombining the terms on the r.h.s. of (\ref{eq:recom}), we obtain 
\begin{equation}\label{eq:lm1-comm} \begin{split} 
\big[ &\cU^* (t;s) e^{i x \cdot p + \varepsilon \nabla \cdot q}   \cU (t;s) , A \big] \\ = \; &C \int_0^\infty d\kappa \, \sqrt{\kappa} \, \left( \frac{\eps (i\nabla - r)}{\kappa + 1 + \eps^2 (i\nabla -r)^2} - \frac{i\eps \nabla}{\kappa + 1 - \eps^2 \Delta} \right) \left[ \cU^* (t;s) e^{i x \cdot p + \varepsilon \nabla \cdot q}   \cU (t;s) , \eps \nabla \right] \frac{1}{\kappa+1+ \eps^2 (i\nabla - r)^2} \\ &+ C \int_0^\infty d\kappa \, \sqrt{\kappa} \, \frac{i\eps \nabla}{\kappa+1 - \eps^2 \Delta} \, \left[ \cU^* (t;s) e^{i x \cdot p + \varepsilon \nabla \cdot q}   \cU (t;s) , \eps \nabla \right] \, \left( \frac{1}{\kappa+1 + \eps^2 (i\nabla - r)^2} - \frac{1}{\kappa+1- \eps^2 \Delta} \right) \\
&+  C \int_0^\infty d\kappa \, \sqrt{\kappa} \, \left( \frac{1}{\kappa + 1 + \eps^2 (i\nabla -r)^2} - \frac{1}{\kappa + 1 - \eps^2 \Delta} \right) \left[ \cU^* (t;s) e^{i x \cdot p + \varepsilon \nabla \cdot q}   \cU (t;s) , \eps \nabla \right] \frac{\eps (i\nabla - r)}{\kappa+1+ \eps^2 (i\nabla - r)^2} \\ &+ C \int_0^\infty d\kappa \, \sqrt{\kappa} \, \frac{1}{\kappa+1 - \eps^2 \Delta} \left[ \cU^* (t;s) e^{i x \cdot p + \varepsilon \nabla \cdot q}   \cU (t;s) , \eps \nabla \right] \, \left( \frac{\eps(i\nabla -r)}{\kappa+1 + \eps^2 (i\nabla - r)^2} - \frac{i\eps \nabla}{\kappa+1- \eps^2 \Delta} \right)
\end{split} \end{equation}
We have the bounds $\| (\kappa+q +\eps^2 (i\nabla -r)^2)^{-1} \| \leq (1+\kappa)^{-1}$, $\| \eps (i\nabla -r)/(\kappa+q +\eps^2 (i\nabla -r)^2) \| \leq (1+\kappa)^{-1/2}$. Furthermore, we easily find 
\[ \left\| \frac{\eps (i\nabla -r)}{\kappa+1 + \eps^2 (i\nabla - r)^2} - \frac{i\eps \nabla}{\kappa+1 -\eps^2 \Delta} \right\|   \leq C \left( \frac{\eps |r|}{\kappa+1} + \frac{\eps^2 r^2}{(\kappa+1)^{3/2}} \right) \leq C \frac{\eps (1 + r^2)}{1+\kappa} \]
and, similarly,  
\[ \left\| \frac{1}{\kappa+1 + \eps^2 (i\nabla - r)^2} - \frac{1}{\kappa+1 -\eps^2 \Delta} \right\|  \leq \frac{C \eps (1+r^2)}{(1+\kappa)^{3/2}} \]
for the norms of the operators appearing in (\ref{eq:lm1-comm}). We can  insert (\ref{eq:lm1-comm}) into (\ref{eq:lm1-A}) and, for each one of the terms appearing on the r.h.s. of (\ref{eq:lm1-comm}), for each $\kappa \in [0 ; \infty)$ and for each time $\tau$, we can replace $\omega$ with a new appropriate trace class operator $\wt{\omega}_{\tau,\kappa}$. For example, for the first term on the r.h.s. of (\ref{eq:lm1-comm}), we define 
\[ \wt{\omega}_{\kappa,\tau} = \frac{C (\kappa+1)^{5/2}}{\eps (1+r^2)} \, \frac{1}{\kappa +1+ \eps^2 (i\nabla - r)^2} \, \cU^* (\tau;0) \, \omega \, \cU^* (\tau; 0) \left( \frac{\eps (i\nabla - r)}{\kappa + 1 + \eps^2 (i\nabla -r)^2} - \frac{i\eps \nabla}{\kappa + 1 - \eps^2 \Delta} \right) \]
where the constant $C > 0$ is chosen so, that $\tr \, |\wt{\omega}_{\kappa,\tau}| \leq \tr \, |\omega| \leq 1$. Integrating over $\kappa$, we arrive at 
\begin{equation}\label{eq:lm1-wtomega} \begin{split} \Big| \tr \, \wt{\cU}^* (s;0) &\left[ e^{ir \cdot x}, \cU^* (t;s) e^{i x \cdot p + \varepsilon \nabla \cdot q}   \cU (t;s) \right] \cU (s;0) \omega \Big| \\ \leq \; &C \eps |r| |q| + C (1+r^2) \int_s^t d\tau \, \sup_{\wt{\omega}} \left| \tr \,  \left[ \cU^* (t;\tau) e^{i x \cdot p + \varepsilon \nabla \cdot q} \cU (t;\tau) , \eps \nabla \right] \wt{\omega} \right| \end{split}
\end{equation} 
where the supremum is taken again over all $\wt{\omega}$ with $\tr  \, |\wt{\omega}| \leq 1$. The estimate (\ref{eq:lm1-wtomega}) holds true for all trace class operators $\omega$ with $\tr \, |\omega| \leq 1$. {F}rom (\ref{eq:lm1-0}), we conclude that 
\begin{equation}\label{eq:lm1-gron1} \begin{split}  \sup_{\omega, r \in \bR^3} \frac{1}{1+r^2} \, \Big| \tr \, & \left[ e^{ir \cdot x}, \cU^* (t;s) e^{i x \cdot p + \varepsilon \nabla \cdot q}   \cU (t;s) \right] \omega \Big| \\ \leq \; &C \eps |q| + C  \int_s^t d\tau \, \sup_{\omega} \left| \tr \,  \left[ \cU^* (t;\tau) e^{i x \cdot p + \varepsilon \nabla \cdot q} \cU (t;\tau) , \eps \nabla \right] \omega \right| \end{split}
\end{equation} 

Similarly as in (\ref{eq:lm1-0}), we observe that 
\begin{equation}\label{eq:lm1-0'} \begin{split} 
\sup_\omega \, \Big| \tr \; \left[ \cU^* (t;s)  e^{i x \cdot p + \varepsilon \nabla \cdot q} \cU (t;s) , \eps \nabla \right] \omega \, \Big| = \; &\sup_\omega  \Big| \tr \; \left[ \cU^* (t;s)  e^{i x \cdot p + \varepsilon \nabla \cdot q} \cU (t;s) , \eps \nabla \right] \cU (s;0) \, \omega \, \cU^* (s;0) \Big|
 \\ = \; & \sup_\omega  \Big| \tr \; \cU^* (s;0) \left[ \cU^* (t;s)  e^{i x \cdot p + \varepsilon \nabla \cdot q} \cU (t;s) , \eps \nabla \right] \cU (s;0) \, \omega \,  \Big| \end{split} \end{equation}
where suprema are taken over all trace class operators $\omega$ with $\tr \, |\omega| \leq 1$. For a fixed $\omega$ (with $\tr \, |\omega| \leq 1$) and a fixed $t \in \bR$, we compute
\[ \begin{split}
i \eps \partial_s \tr \; \cU^* (s;0) & \left[ \cU^* (t;s)  e^{i x \cdot p + \varepsilon \nabla \cdot q} \cU (t;s) , \eps \nabla \right] \cU (s;0) \, \omega  \\ & \hspace{2cm} =  \tr \; \cU^* (s;0) \left[ \cU^* (t;s)  e^{i x \cdot p + \varepsilon \nabla \cdot q} \cU (t;s) , \left[ h_H (s) ,  \eps \nabla \right] \right] \cU (s;0) \, \omega
\end{split} \]
Since
\[ \big[ h_H (s) , \eps \nabla \big] = \eps \nabla (V * \rho_s) (x) = \int dk \, \eps k \, \widehat{V} (k) 
\, \widehat{\rho}_s (k) e^{ik \cdot x} \]
we obtain, integrating from $s$ to $t$, 
\[ \begin{split} 
\tr \; \cU^* (s;0) & \left[ \cU^* (t;s)  e^{i x \cdot p + \varepsilon \nabla \cdot q} \cU (t;s) , \eps \nabla \right] \cU (s;0) \, \omega   \\ = \; & \tr \; \cU^* (t;0) \left[ e^{i x \cdot p + \varepsilon \nabla \cdot q} , \eps \nabla \right] \cU (t;0) \, \omega \\ &-i \int_s^t d\tau \int dk \, k \, \widehat{V} (k) \widehat{\rho}_s (k)  \, \tr \; \cU^* (\tau;0) \left[ \cU^* (t;\tau)  e^{i x \cdot p + \varepsilon \nabla \cdot q} \cU (t;\tau) , e^{ik \cdot x} \right] \cU (\tau;0) \, \omega \end{split} \]
With 
\[ \left[ e^{i x \cdot p + \varepsilon \nabla \cdot q} , \eps \nabla \right] = \eps p \, e^{i x \cdot p + \varepsilon \nabla \cdot q} \]
and with the assumption (\ref{ass:lemma_smallr}), we conclude that 
\[ \begin{split} 
\Big| \, \tr \; \cU^* (s;0) &\left[ \cU^* (t;s)  e^{i x \cdot p + \varepsilon \nabla \cdot q} \cU (t;s) , \eps \nabla \right] \cU (s;0) \, \omega \Big| \\ & \hspace{2cm} \leq \eps |p| + C \int_s^t  d\tau \sup_{\omega,k \in \bR^3} \frac{1}{1+k^2} \, \left| \tr \;  \left[ \cU^* (t;\tau)  e^{i x \cdot p + \varepsilon \nabla \cdot q} \cU (t;\tau) , e^{ik \cdot x} \right] \omega \right| \end{split} \]
By (\ref{eq:lm1-0'}), we find
\begin{equation}\label{eq:lm1-gron2}
\begin{split} 
\sup_\omega & \left| \tr \; \left[ \cU^* (t;s)  e^{i x \cdot p + \varepsilon \nabla \cdot q} \cU (t;s) , \eps \nabla \right] \omega \right| \\ & \hspace{1cm} \leq \eps |p| + C \int_s^t d\tau \sup_{\omega, k \in \bR^3} \frac{1}{1+k^2} \left| \tr \;  \left[ \cU^* (t;\tau)  e^{i x \cdot p + \varepsilon \nabla \cdot q} \cU (t;\tau) , e^{ik \cdot x} \right] \omega \right| \end{split} \end{equation}
where, as usual, suprema are taken over trace class operators $\omega$, with $\tr \, |\omega| \leq 1$. 
Combining (\ref{eq:lm1-gron1}) with (\ref{eq:lm1-gron2}) and applying Gronwall's lemma, we arrive at
\[ \sup_{\omega, k \in \bR^3} \frac{1}{1+k^2} \left| \tr \;  \left[ \cU^* (t;s)  e^{i x \cdot p + \varepsilon \nabla \cdot q} \cU (t;s) , e^{ik \cdot x} \right] \omega \right| \leq C (|p|+|q|) \eps e^{C|t-s|} \]
and
\[ \sup_{\omega} \left| \tr \;  \left[ \cU^* (t;s)  e^{i x \cdot p + \varepsilon \nabla \cdot q} \cU (t;s) , \eps \nabla \right] \omega \right| \leq C (|p|+|q|) \eps e^{C|t-s|} \]
This completes the proof of the lemma.
\end{proof}

We also need bounds involving second commutators of evolved observables of the form $\cU^* (t;s) \exp (i x \cdot p + \eps \nabla \cdot q) \cU (t;s)$ with position or momentum operators; these are shown in the next lemma.
\begin{lemma}
\label{lemma:double}
Assume that $V \in L^1\left( \mathbb{R}^3\right)$ satisfies
\begin{equation}\label{eq:lm2-assV}
\int d p \; ( 1+|p|^4)  | \widehat{V}(p) | < \infty. 
\end{equation}
Then there exist $C>0$ such that
\[ \begin{split} 
\sup_{i,j, \omega} \Big| \tr \left[ x_i, \left[ x_j, \mathcal{U}^* (t;s) e^{i x \cdot p + \varepsilon \nabla \cdot q}\mathcal{U}(t;s) \right] \right] \omega \, \Big| &\leq C \varepsilon^2 (|p|+|q|) (1+|p|+|q|) e^{C|t-s|},\\
\sup_{i,j, \omega} \Big| \tr \left[ \eps \nabla_i, \left[ x_j,\mathcal{U}^*(t;s)e^{i x \cdot p + \varepsilon \nabla \cdot q}\mathcal{U}(t;s) \right] \right] \omega \, \Big| &\leq C \varepsilon^2 (|p|+|q|) (1+|p|+|q|)  e^{C|t-s|}, \\
\sup_{i,j, \omega} \Big| \tr \left[ \eps \nabla_i, \left[ \eps \nabla_j, \mathcal{U}^*(t;s)e^{i x \cdot p + \varepsilon \nabla \cdot q}\mathcal{U}(t;s) \right] \right] \omega \, \Big| &\leq C \varepsilon^2 (|p|+|q|) (1+|p|+|q|) e^{C|t-s|},
\end{split} \]
where suprema are taken over all $i,j = 1,2,3$ and over all trace class operators $\omega$ on $L^2 \left( \mathbb{R}^3 \right) $ with $\tr \, |\omega| \leq 1$. 
\end{lemma}

\begin{proof}
Recall that $\cU (t;s)$ is the dynamics generated by the time-dependent Hartree Hamiltonian $h_H (t) = \sqrt{1-\eps^2 \Delta} + (V*\rho_t) (x)$, i.e.
\[ i\eps \partial_t \cU (t;s) = h_H (t) \cU (t;s) \]
with $\cU (s;s) = 1$ for all $s \in \bR$. To prove the lemma, we start by noticing that
\begin{equation}\label{eq:lm2-step1-0} \begin{split} \sup_{i,j,\omega} \Big| \tr \; &\left[ x_i , \left[ x_j, \cU^* (t;s) e^{ip \cdot x + \eps \nabla \cdot q} \cU (t;s) \right] \right] \omega \Big|  \\ = \; &  \sup_{i,j,\omega} \Big| \tr \;  \left[ x_i , \left[ x_j, \cU^* (t;s) e^{ip \cdot x + \eps \nabla \cdot q} \cU (t;s) \right] \right] \cU(s;0) \, \omega \, \cU^* (s;0)  \Big| \\ = \; &  \sup_{i,j,\omega} \Big| \tr \; \cU^* (s;0)  \left[ x_i , \left[ x_j, \cU^* (t;s) e^{ip \cdot x + \eps \nabla \cdot q} \cU (t;s) \right] \right] \cU(s;0) \, \omega  \Big| \end{split} \end{equation}
For fixed $\omega$ with $\tr \, |\omega| \leq 1$, and for fixed $t \in \bR$, we compute
\[ \begin{split} 
i\eps \partial_s  \tr \; \cU^* (s;0) &\left[ x_i , \left[ x_j , \cU^* (t;s) e^{ip \cdot x + \eps \nabla \cdot q} \cU (t;s) \right] \right] \cU (s;0) \, \omega  \\ = \; &\tr \; \cU^* (s;0) \left[ h_H (s), \left[ x_i, \left[ x_j , \cU^* (t;s) e^{ip \cdot x + \eps \nabla \cdot q} \cU (t;s) \right] \right] \right] \cU (s;0) \, \omega \\ &- \tr \; \cU^* (s;0) \left[ x_i, \left[ x_j ,  \left[ h_H (s), \cU^* (t;s) e^{ip \cdot x + \eps \nabla \cdot q} \cU (t;s) \right] \right] \right] \cU (s;0) \, \omega \\
= \; & \tr \; \cU^* (s;0) \left[ \left[ h_H (s) , x_i \right] , \left[ x_j, \cU^* (t;s) e^{ip \cdot x + \eps \nabla \cdot q} \cU (t;s) \right] \right] \cU (s;0) \, \omega \\ &+   \tr \; \cU^* (s;0) \left[  x_i , \left[ \left[ h_H (s), x_j \right] , \cU^* (t;s) e^{ip \cdot x + \eps \nabla \cdot q} \cU (t;s) \right] \right] \cU (s;0) \, \omega \end{split} \]
With
\begin{equation}\label{eq:comm-hx} \left[ h_H (s) , x_i \right] = \frac{i \eps^2 \nabla_i}{\sqrt{1-\eps^2 \Delta}} = \eps \frac{i\eps \nabla_i}{\sqrt{1-\eps^2 \Delta}} \end{equation}
we obtain  
\[\begin{split}  i \eps \partial_s \tr \; \cU^* (s;0) &\left[ x_i , \left[ x_j ,  \cU^* (t;s) e^{ip \cdot x + \eps \nabla \cdot q} \cU (t;s) \right]  \right] \cU (s;0) \, \omega \\ = \; & \eps \, \tr \; \cU^* (s;0) \, \left[ \frac{i\eps \nabla_i}{\sqrt{1-\eps^2 \Delta}}, \left[ x_j,  \cU^* (t;s) e^{ip \cdot x + \eps \nabla \cdot q} \cU (t;s) \right]  \right] \cU (s;0) \, \omega \\ &+ \eps \, \tr \; \cU^* (s;0) \, \left[ \frac{i\eps \nabla_j}{\sqrt{1-\eps^2 \Delta}}, \left[ x_i,  \cU^* (t;s) e^{ip \cdot x + \eps \nabla \cdot q} \cU (t;s) \right]  \right] \cU (s;0) \, \omega \\
&+ \eps \, \tr \; \cU^* (s;0) \left[ \left[ \frac{i\eps \nabla_j}{\sqrt{1-\eps^2 \Delta}} , x_i \right] , \cU^* (t;s) e^{ip \cdot x + \eps \nabla \cdot q} \cU (t;s) \right] \cU (s;0) \, \omega \end{split} \]
Integrating from $s$ to $t$, we find 
\begin{equation}\label{eq:lm2-step1} \begin{split} 
\tr \; \cU^* (s;0) &\left[ x_i , \left[ x_j , \cU^* (t;s) e^{ip \cdot x + \eps \nabla \cdot q} \cU (t;s) \right] \right] \cU (s;0) \, \omega \\ = \; & \tr \; \cU^* (t;0) \left[ x_i, \left[ x_j , e^{ip \cdot x + \eps \nabla \cdot q} \right]  \right] \cU (t;0) \, \omega \\ &-i \int_s^t d\tau \, \tr \; \cU^* (\tau ;0) \left[ \frac{i\eps \nabla_i}{\sqrt{1-\eps^2 \Delta}}, \left[ x_j, \cU^* (t;\tau) e^{ip \cdot x + \eps \nabla \cdot q} \cU (t;\tau) \right] \right] \cU(\tau;0) \, \omega \\ &-i \int_s^t d\tau \, \tr \; \cU^* (\tau ;0) \left[ \frac{i\eps \nabla_j}{\sqrt{1-\eps^2 \Delta}}, \left[ x_i, \cU^* (t;\tau) e^{ip \cdot x + \eps \nabla \cdot q} \cU (t;\tau) \right] \right] \cU(\tau;0) \, \omega \\
&-i \int_s^t d\tau \, \tr \; \cU^* (\tau ;0) \left[ \left[ \frac{i\eps \nabla_j}{\sqrt{1-\eps^2 \Delta}}, x_i \right] , \cU^* (t;\tau) e^{ip \cdot x + \eps \nabla \cdot q} \cU (t;\tau) \right] \cU (\tau ;0) \, \omega
\end{split} \end{equation}
To estimate the first term on the r.h.s. of (\ref{eq:lm2-step1}), we notice that
\[ \left[ x_i , \left[ x_j , e^{ip \cdot x + \eps \nabla \cdot q} \right] \right] = \eps^2 q_i q_j e^{ip \cdot x + \eps \nabla \cdot q}  \]
Hence
\begin{equation}\label{eq:lm2-step1-1} \Big| \tr \; \cU^* (t;0)   \left[ x_i , \left[ x_j , e^{ip \cdot x + \eps \nabla \cdot q} \cU (t;\tau) \right] \right] \cU(t;0) \omega \Big| \leq \eps^2 q^2 \end{equation}
for all trace class operators $\omega$ with $\tr \, |\omega| \leq 1$ and for all $i,j = 1,2,3$. 

To bound the second term on the r.h.s. of (\ref{eq:lm2-step1}), we write
\[ \frac{\eps \nabla_i}{\sqrt{1-\eps^2 \Delta}} = C \int_0^\infty d\kappa \, \sqrt{\kappa} \frac{\eps \nabla_i}{(\kappa+1-\eps^2 \Delta)^2} \]
for an appropriate constant $C \in \mathbb{C}$. Therefore, we have
\begin{equation}\label{eq:lm2-step12} \begin{split} 
\tr\; &\cU^* (\tau;0) \left[ \frac{i\eps \nabla_i}{\sqrt{1-\eps^2\Delta}} , \left[ x_j , \cU^* (t;\tau) e^{ip \cdot x + \eps \nabla \cdot q}  \cU(t;\tau) \right] \right] \cU (\tau;0) \, \omega \\
= \; &C \int_0^\infty d\kappa \, \sqrt{\kappa} \, \tr \; \cU^* (\tau;0) \frac{1}{\kappa + 1 -\eps^2 \Delta} \left[ \eps^2 \Delta , \left[ x_j,  \cU^* (t;\tau) e^{ip \cdot x + \eps \nabla \cdot q}  \cU(t;\tau) \right] \right] \frac{\eps \nabla_i}{(\kappa+1 - \eps^2 \Delta)^2} \, \cU (\tau ;0) \, \omega \\
&+ C \int_0^\infty d\kappa \, \sqrt{\kappa} \, \tr \; \cU^* (\tau;0) \frac{1}{\kappa + 1 -\eps^2 \Delta} \left[ \eps \nabla_i , \left[ x_j,  \cU^* (t;\tau) e^{ip \cdot x + \eps \nabla \cdot q}  \cU(t;\tau) \right] \right] \frac{1}{\kappa+1 - \eps^2 \Delta} \, \cU (\tau ;0) \, \omega \\
&+ C \int_0^\infty d\kappa \, \sqrt{\kappa} \, \tr \; \cU^* (\tau;0) \frac{\eps \nabla_i}{(\kappa + 1 -\eps^2 \Delta)^2} \left[ \eps^2 \Delta , \left[ x_j,  \cU^* (t;\tau) e^{ip \cdot x + \eps \nabla \cdot q}  \cU(t;\tau) \right] \right] \frac{1}{\kappa+1 - \eps^2 \Delta} \, \cU (\tau ;0) \, \omega \end{split} \end{equation}
Expanding
\[ \begin{split} &\left[ \eps^2 \Delta , \left[ x_j , \cU^* (t;\tau) e^{ip \cdot x + \eps \nabla \cdot q}  \cU(t;\tau) \right] \right] \\ & \hspace{1cm}  = \sum_{m=1}^3 \left\{ \eps \nabla_m \left[ \eps \nabla_m , \left[ x_j , \cU^* (t;\tau) e^{ip \cdot x + \eps \nabla \cdot q}  \cU(t;\tau) \right] \right] + \left[ \eps \nabla_m , \left[ x_j , \cU^* (t;\tau) e^{ip \cdot x + \eps \nabla \cdot q}  \cU(t;\tau) \right] \right] \eps \nabla_m \right\} \end{split} 
\]
in the first and in the third term on the r.h.s. of (\ref{eq:lm2-step12}) and using everywhere the bounds
\[ \begin{split} \left\| \frac{1}{(\kappa + 1 - \eps^2 \Delta)^{n}} \right\| &\leq \frac{1}{(\kappa +1)^{n}} ,
\qquad \left\| \frac{\eps \nabla_i}{(\kappa + 1 - \eps^2 \Delta)^{n}} \right\| \leq \frac{C}{(\kappa +1)^{n-1/2}} , \qquad  \left\| \frac{\eps^2 \, \nabla_i \, \nabla_j}{(\kappa + 1 - \eps^2 \Delta)^{n}} \right\| \leq \frac{C}{(\kappa +1)^{n-1}} \end{split} \]
for $n=1,2$ and $i,j =1,2,3$, we conclude that, for any $\omega$ with $\tr \, |\omega| \leq 1$, 
\begin{equation}\label{eq:lm2-step1-2}
\begin{split} \Big| \tr \; \cU^* (\tau;0) &\left[ \frac{i\eps \nabla_i}{\sqrt{1-\eps^2 \Delta}} , \left[ x_j, \cU^* (t;\tau) e^{ip \cdot x + \eps \nabla \cdot q}  \cU(t;\tau) \right] \right] \cU(\tau;0) \omega \Big|  \\ &\hspace{3cm} \leq C \sup_{i,j,\wt{\omega}} \left| \tr \; \left[ \eps \nabla_i , \left[ x_j , \cU^* (t;\tau) e^{ip \cdot x + \eps \nabla \cdot q}  \cU(t;\tau) \right] \right] \wt{\omega} \right|
\end{split} 
\end{equation}
where the supremum is taken over all $\wt{\omega}$ with $\tr \, |\wt{\omega}| \leq 1$. 

The third term on the r.h.s. of (\ref{eq:lm2-step1}) can be estimated analogously. As for the fourth term, we can proceed similarly, using the representation 
\[ \left[ x_i , \frac{i\eps \nabla_j}{\sqrt{1-\eps^2 \Delta}} \right] = C \eps \delta_{ij} \int_0^\infty d\kappa \, \sqrt{\kappa} \, \frac{1}{(\kappa+1-\eps^2 \Delta)^2} + C \eps \int_0^\infty d\kappa \, \sqrt{\kappa} \, \frac{\eps \nabla_i \eps \nabla_j}{(\kappa+1-\eps^2 \Delta)^3} \]
We find, using Lemma \ref{lemma:smallr},  
\begin{equation}\label{eq:lm2-step1-4} 
\begin{split}  \Big| \tr \; \cU^* (\tau;0) &\left[ \left[ x_i , \frac{i\eps \nabla_j}{\sqrt{1-\eps^2 \Delta}} \right] , \cU^* (t;\tau) e^{ip \cdot x + \eps \nabla \cdot q}  \cU(t;\tau) \right]  \cU (\tau;0) \, \omega \Big| \\ &\hspace{2cm} \leq C \eps \sup_{i, \omega} \left| \tr \left[ \eps \nabla_i , \cU^* (t;\tau) e^{ip \cdot x + \eps \nabla \cdot q}  \cU(t;\tau) \right]  \omega \right| \leq C \eps^2 (|p| + |q|) e^{C|t-\tau|} \end{split} \end{equation}
 
Combining (\ref{eq:lm2-step1-1}), (\ref{eq:lm2-step1-2}) and (\ref{eq:lm2-step1-4}) we conclude from (\ref{eq:lm2-step1}) that
\[ \begin{split} \Big| \tr \; \cU^* (s;0) &\left[ x_i , \left[ x_j, \cU^* (t;\tau) e^{ip \cdot x + \eps \nabla \cdot q}  \cU(t;\tau) \right] \right] \cU (s;0)  \omega \Big| \\ \leq \; & \eps^2 |q|^2 + C \eps^2 (|p| + |q|)  e^{C |t-s|} + C \int_s^t d\tau \, \sup_{i,j,\omega} \left| \tr \left[ \eps \nabla_i , \left[ x_j , \cU^* (t;\tau) e^{ip \cdot x + \eps \nabla \cdot q}  \cU(t;\tau)  \right] \right] \omega \right|   \end{split} \]
{F}rom (\ref{eq:lm2-step1-0}), we find
\begin{equation}\label{eq:lm2-step1-f}
\begin{split} \sup_{i,j,\omega} \Big| &\tr \; \cU^* (s;0) \left[ x_i , \left[ x_j, \cU^* (t;\tau) e^{ip \cdot x + \eps \nabla \cdot q}  \cU(t;\tau) \right] \right] \cU (s;0)  \omega \Big| \\ \leq \; & C \eps^2 (|p|+|q|) (|p| + |q| +1)  e^{C |t-s|} + C \int_s^t d\tau \, \sup_{i,j,\omega} \left| \tr \left[ \eps \nabla_i , \left[ x_j , \cU^* (t;\tau) e^{ip \cdot x + \eps \nabla \cdot q}  \cU(t;\tau)  \right] \right] \omega \right|   \end{split} \end{equation}
where suprema are taken over all $i,j =1,2,3$ and all trace class operators $\omega$ with $\tr \, |\omega| \leq 1$. 

Next, we bound the growth of the integrand on the r.h.s. of (\ref{eq:lm2-step1-f}); the goal is to close the estimates and apply Gronwall's lemma. As above, we observe that
\begin{equation}\label{eq:lm2-step2-0} 
\begin{split} 
\sup_{\omega,i,j} \Big| \tr \, &\left[ \eps \nabla_i, \left[ x_j , \cU^* (t;s) e^{ip \cdot x + \eps \nabla \cdot q}  \cU(t;s)  \right] \right] \omega \Big| \\ &\hspace{3cm} = \sup_{\omega,i,j} \Big| \tr \, \cU^* (s;0) \left[ \eps \nabla_i, \left[ x_j , \cU^* (t;s) e^{ip \cdot x + \eps \nabla \cdot q}  \cU(t;s)  \right]\right] \cU (s;0) \, \omega \Big| \end{split} \end{equation}
For fixed $\omega$ with $\tr \, |\omega| \leq 1$ and fixed $t \in \bR$, we compute 
\[ \begin{split} 
i\eps \partial_s \tr \; \cU^* (s;0) &\left[ \eps \nabla_i , \left[ x_j , \cU^* (t;s) e^{ip \cdot x + \eps \nabla \cdot q}  \cU(t;s)  \right] \right] \cU (s;0) \omega  \\
= \; &\tr \; \cU^* (s;0) \left[ \left[ h_H(s) , \eps \nabla_i \right], \left[ x_j,  \cU^* (t;s) e^{ip \cdot x + \eps \nabla \cdot q}  \cU(t;s)  \right] \right] \cU (s;0) \, \omega \\ &+ \tr \; \cU^* (s;0) \left[ \eps \nabla_i , \left[ \left[ h_H(s) , x_j \right],  \cU^* (t;s) e^{ip \cdot x + \eps \nabla \cdot q}  \cU(t;s)  \right] \right] \cU (s;0) \, \omega 
\end{split} \]
With
\begin{equation}\label{eq:comm-hnab}
\left[ h_H(s) , \eps \nabla_i \right] = \eps \nabla (V* \rho_s) (x) = \eps \int dk \, k \widehat{V} (k) \widehat{\rho}_s (k) e^{ik \cdot x} \end{equation}
and (\ref{eq:comm-hx}), we obtain, integrating from $s$ to $t$,  
\begin{equation}\label{eq:lm2-step2-1} \begin{split} 
\tr \; \cU^* (s;0) &\left[ \eps \nabla_i, \left[ x_j ,  \cU^* (t;s) e^{ip \cdot x + \eps \nabla \cdot q}  \cU(t;s)  \right] \right] \cU (s;0) \omega  \\ = \;&
\tr \; \cU^* (t;0) \left[ \eps \nabla_i , \left[ x_j ,  e^{ip \cdot x + \eps \nabla \cdot q} \right] \right] \cU (t;0) \, \omega \\ &-i \int_s^t d\tau \int dk \, k \, \widehat{V} (k) \widehat{\rho}_\tau (k) \, \tr\; \cU^* (\tau;0) \left[ e^{ik \cdot x}, \left[ x_j ,  \cU^* (t;\tau) e^{ip \cdot x + \eps \nabla \cdot q}  \cU(t;\tau) \right] \right] \cU (\tau;0) \, \omega \\
 &-i \int_s^t d\tau \, \tr \; \cU^* (\tau;0) \left[ \frac{\eps \nabla_j}{\sqrt{1-\eps^2 \Delta}} , \left[ \eps \nabla_i ,  \cU^* (t;\tau) e^{ip \cdot x + \eps \nabla \cdot q}  \cU(t;\tau) \right]  \right] \cU (\tau;0) \, \omega  \end{split} \end{equation}
Using 
\[ \left[ \eps \nabla_i, \left[ x_j ,  e^{ip \cdot x + \eps \nabla \cdot q}  \right] \right] = \eps^2 p_i q_j e^{ip \cdot x + \eps \nabla \cdot q} \, , \]
the representation
\[ \left[ e^{ik \cdot x}, \left[ x_j, \cU^* (t;\tau) e^{ip \cdot x + \eps \nabla \cdot q}  \cU(t;\tau) \right] \right] = \int_0^1 d\kappa \, e^{i\kappa k \cdot x} \left[ ik \cdot x , \left[ x_j, \cU^* (t;\tau) e^{ip \cdot x + \eps \nabla \cdot q}  \cU(t;\tau) \right] \right] e^{i (1-\kappa) k \cdot x} \, , \]
the assumption (\ref{eq:lm2-assV}) on the potential $V$ and proceeding as in (\ref{eq:lm2-step12}) to control the second term on the r.h.s. of (\ref{eq:lm2-step2-1}), we arrive at
\begin{equation}\label{eq:lm2-step2-f} \begin{split} 
\sup_{i,j, \omega} \Big| \tr \; &\left[ \eps \nabla_i , \left[ x_j ,  \cU^* (t;s) e^{ip \cdot x + \eps \nabla \cdot q}  \cU(t;s) \right] \right] \omega \Big| \\ \leq \; &\eps^2 |p| |q| + \int_s^t d\tau \sup_{i,j,\omega} \left| \tr \; \left[ x_i, \left[ x_j ,  \cU^* (t;\tau) e^{ip \cdot x + \eps \nabla \cdot q}  \cU(t;\tau) \right] \right] \omega \right| \\ &+ \int_s^t d\tau \sup_{i,j,\omega} \left| \tr \; \left[ \eps \nabla_i, \left[ \eps \nabla_j ,  \cU^* (t;\tau) e^{ip \cdot x + \eps \nabla \cdot q}  \cU(t;\tau) \right] \right] \omega \right|  \end{split} \end{equation}
where suprema are taken over $i,j =1,2,3$ and all $\omega$ with $\tr \, |\omega| \leq 1$. 

To close the estimate, we have to control the growth of the last integrand on the r.h.s. of (\ref{eq:lm2-step2-f}). To this end, we notice first of all that
\begin{equation}\label{eq:lm2-step3-0} \begin{split} 
\sup_{i,j,\omega} \Big| \tr\; &\left[ \eps \nabla_i , \left[  \eps \nabla_j ,  \cU^* (t;s) e^{ip \cdot x + \eps \nabla \cdot q}  \cU(t;s) \right] \right] \omega \Big| \\ &\hspace{2cm} = \sup_{i,j,\omega} \Big| \tr \; \cU^* (s;0) \left[ \eps \nabla_i , \left[  \eps \nabla_j ,  \cU^* (t;s) e^{ip \cdot x + \eps \nabla \cdot q}  \cU(t;s) \right] \right] \cU (s;0) \, \omega \Big| \end{split} \end{equation}
For fixed $\omega$ with $\tr \, |\omega| \leq 1$ and $t \in \bR$, we have
\[ \begin{split} i\eps \partial_s \tr \; \cU^* (s;0) &\left[ \eps \nabla_i, \left[ \eps \nabla_j , \cU^* (t;s) e^{ip \cdot x + \eps \nabla \cdot q}  \cU(t;s) \right] \right] \cU(s;0) \, \omega  \\ = \; &\tr \; \cU^* (s;0) \left[ \left[ h_H (s) , \eps \nabla_i \right], \left[ \eps \nabla_j ,  \cU^* (t;s) e^{ip \cdot x + \eps \nabla \cdot q}  \cU(t;s) \right] \right] \cU (s;0)\, \omega \\ &+ \tr \; \cU^* (s;0) \left[ \eps \nabla_i, \left[ \left[ h_H (s), \eps \nabla_j \right] , \cU^* (t;s) e^{ip \cdot x + \eps \nabla \cdot q}  \cU(t;s) \right] \right] \cU (s;0) \; \omega \end{split} \]
With (\ref{eq:comm-hnab}) we find
\begin{equation}\label{eq:lm2-step3-1} \begin{split} 
\tr \; \cU^* (s;0) &\left[ \eps \nabla_i, \left[ \eps \nabla_j , \cU^* (t;s) e^{ip \cdot x + \eps \nabla \cdot q}  \cU(t;s) \right] \right] \cU (s;0) \, \omega \\
= \; & \tr \; \cU^* (t;0) \left[ \eps \nabla_i, \left[ \eps \nabla_j, e^{ip \cdot x + \eps \nabla \cdot q} \right] \right] \cU (t;0) \, \omega \\ &-i \int_s^t d\tau \int dk \, k_i \, \widehat{V} (k)  \widehat{\rho}_\tau (k) \, \tr \; \cU^* (\tau;0) \left[ e^{ik \cdot x} , \left[ \eps \nabla_j, \cU^* (t;\tau) e^{ip \cdot x + \eps \nabla \cdot q}  \cU(t;\tau) \right] \right] \cU (\tau;0) \, \omega \\
&-i\int_s^t d\tau \, \int dk \, k_j \, \widehat{V} (k) \widehat{\rho}_\tau (k) \, \tr \; \cU^* (\tau ;0) \left[ \eps \nabla_i, \left[ e^{ik\cdot x}, \cU^* (t;\tau) e^{ip \cdot x + \eps \nabla \cdot q}  \cU(t;\tau) \right] \right] \cU (\tau;0) \,\omega \end{split} \end{equation}
Using
\[ \left[ \eps \nabla_i, \left[ \eps \nabla_j ,  e^{ip \cdot x + \eps \nabla \cdot q} \right] \right] = \eps^2 p_i p_j e^{ip \cdot x + \eps \nabla \cdot q} \]
the first term on the r.h.s. of (\ref{eq:lm2-step3-1}) can be bounded by
\begin{equation}\label{eq:lm2-step3-11} \Big| \tr \; \cU^* (t;0) \left[ \eps \nabla_i, \left[ \eps \nabla_j, e^{ip \cdot x + \eps \nabla \cdot q} \right] \right] \cU (t;0) \, \omega  \Big| \leq \eps^2 |p|^2 \end{equation}
for all $\omega$ with $\tr \, |\omega| \leq 1$. With the representation
\[ \left[ e^{ik \cdot x}, \left[ \eps \nabla_i , \cU^* (t;\tau) e^{ip \cdot x + \eps \nabla \cdot q}  \cU(t;\tau) \right] \right] = \int_0^1 d\kappa \, e^{i\kappa k\cdot x} \left[ ik\cdot x, \left[ \eps \nabla_i , \cU^* (t;\tau) e^{ip \cdot x + \eps \nabla \cdot q}  \cU(t;\tau) \right] \right] e^{i(1-\kappa) k\cdot x} \]
we can bound the integrand in the second term on the r.h.s. of (\ref{eq:lm2-step3-1}) by
\begin{equation}\label{eq:lm2-step3-12} \begin{split} 
\Big|  \tr \; \cU^* (\tau;0) &\left[ e^{ik \cdot x} , \left[ \eps \nabla_j, \cU^* (t;\tau) e^{ip \cdot x + \eps \nabla \cdot q}  \cU(t;\tau) \right] \right] \cU (\tau;0) \, \omega \Big| \\
&\leq C |k| \sup_{i,j,\omega} \Big| \tr\; \left[ \eps \nabla_j, \left[ x_i , \cU^* (t;\tau) e^{ip \cdot x + \eps \nabla \cdot q}  \cU(t;\tau) \right] \right] \omega \Big| \end{split}  \end{equation}
where the supremum is taken over all $i,j =1,2,3$ and all $\omega$ with $\tr \, |\omega| \leq 1$ (here we  used the fact that $[x_i, [ \eps \nabla_j , A]] = [ \eps, \nabla_j, [x_i, A]]$ for every $A$, since $[\eps \nabla_j, x_i ] = \eps \delta_{ij}$ commutes with all operators). Writing instead 
\[ \begin{split} 
&\left[ \eps \nabla_i , \left[ e^{ik\cdot x},  \cU^* (t;\tau) e^{ip \cdot x + \eps \nabla \cdot q}  \cU(t;\tau) \right] \right] \\ &\hspace{2cm} = \left[ e^{ik\cdot x} , \left[ \eps \nabla_i , \cU^* (t;\tau) e^{ip \cdot x + \eps \nabla \cdot q}  \cU(t;\tau) \right] \right] + i\eps k_i  \left[ e^{ik \cdot x},  \cU^* (t;\tau) e^{ip \cdot x + \eps \nabla \cdot q}  \cU(t;\tau) \right] \end{split} \]
and using Lemma \ref{lemma:smallr}, the integrand in the third term on the r.h.s. of (\ref{eq:lm2-step3-1}) can be estimated by 
\begin{equation} \label{eq:lm2-step3-13} 
\begin{split} 
\Big| \tr \; \cU^* (\tau ;0) &\left[ \eps \nabla_i, \left[ e^{ik\cdot x}, \cU^* (t;\tau) e^{ip \cdot x + \eps \nabla \cdot q}  \cU(t;\tau) \right] \right] \cU (\tau;0) \omega \Big| \\ &\hspace{2cm} \leq C |k| \sup_{i,j,\omega} \left| \tr \; \left[ \eps \nabla_i, \left[ x_j , \cU^* (t;\tau) e^{ip \cdot x + \eps \nabla \cdot q}  \cU(t;\tau) \right] \right] \omega \right| 
+ C \eps^2 |k| (1 + k^2) e^{C|t-\tau|} \end{split} 
\end{equation}
Inserting (\ref{eq:lm2-step3-11}), (\ref{eq:lm2-step3-12}) and (\ref{eq:lm2-step3-13}) in (\ref{eq:lm2-step3-1}) we conclude, with the assumption (\ref{eq:lm2-assV}), that
\[ \begin{split} 
\Big| \tr \; \cU^* (s;0) &\left[ \eps \nabla_i, \left[ \eps \nabla_j , \cU^* (t;s) e^{ip \cdot x + \eps \nabla \cdot q}  \cU(t;s) \right] \right] \cU (s;0) \, \omega \Big| \\ & \hspace{1cm} \leq \eps^2 |p|^2 + C \int_s^t d\tau \, \sup_{i,j,\omega} \left| \tr \;  \left[ \eps \nabla_i, \left[ x_j , \cU^* (t;\tau) e^{ip \cdot x + \eps \nabla \cdot q}  \cU(t;\tau) \right] \right] \omega \right| \end{split} \]
Hence, we arrive at
\[ \begin{split}  \sup_{i,j,\omega} \Big| \tr \;  &\left[ \eps \nabla_i, \left[ \eps \nabla_j , \cU^* (t;s) e^{ip \cdot x + \eps \nabla \cdot q}  \cU(t;s) \right] \right] \, \omega \Big| \\ & \hspace{2cm} \leq \eps^2 |p|^2 + C \int_s^t d\tau \, \sup_{i,j,\omega} \left| \tr \;  \left[ \eps \nabla_i, \left[ x_j , \cU^* (t;\tau) e^{ip \cdot x + \eps \nabla \cdot q}  \cU(t;\tau) \right] \right] \omega \right| \end{split} \]
Combining the last bound with (\ref{eq:lm2-step1-f}) and (\ref{eq:lm2-step2-f}) and applying Gronwall's lemma, we obtain
\[ \begin{split} 
\sup_{i,j,\omega} \Big| \tr \;  \left[ \eps \nabla_i, \left[ \eps \nabla_j , \cU^* (t;s) e^{ip \cdot x + \eps \nabla \cdot q}  \cU(t;s) \right] \right] \, \omega \Big| & \leq C \eps^2 (|p|+|q|)  (1+|p|+|q|) e^{C|t-s|} \\
\sup_{i,j,\omega} \Big| \tr \;  \left[ \eps \nabla_i, \left[ x_j , \cU^* (t;s) e^{ip \cdot x + \eps \nabla \cdot q}  \cU(t;s) \right] \right] \, \omega \Big| & \leq C \eps^2 (|p|+|q|) (1+|p|+|q|) e^{C|t-s|} \\
\sup_{i,j,\omega} \Big| \tr \;  \left[ x_i, \left[ x_j , \cU^* (t;s) e^{ip \cdot x + \eps \nabla \cdot q}  \cU(t;s) \right] \right] \, \omega \Big| & \leq C \eps^2 (|p|+|q|) (1+|p|+|q|) e^{C|t-s|} 
\end{split} \]
as claimed.
\end{proof}

The propagation of regularity for solutions of the relativistic Vlasov equation plays an important role in our analysis. 

\begin{proposition}
\label{prop:regest}
Assume that 
\begin{align}
\int d p \; | \widehat{V}(p) | ( 1+|p|^2) < \infty.
\end{align}
Let $W_t$ be the solution of the Vlasov equation \eqref{eq:Vlasov} with initial data $W_0$. For $k \in \{ 1,2,3,4,5,6\}$, there exists a constant $C>0$, that depends an $\|W_0\|_{H^2_4}$ but not on the higher Sobolev norms, such that
\begin{equation}\label{eq:propreg}
\|W_t\|_{H^k_2}\leq C e^{C |t|} \|W_0\|_{H_2^k}.
\end{equation}
\end{proposition}
\begin{proof}
As explained in the proof of \cite[Prop. B.1]{BPSS}, propagation of regularity follows if we can establish regularity of the flow $(x,v) \to (X_t (x,v), V_t (x,v))$ defined by Newton's equations
\[ \dot{X}_t (x,v) = \frac{V_t (x,v)}{\sqrt{1+ V^2_t (x,v)}}, \qquad \dot{V}_t (x,v) = - \nabla \left( V * \rho_t \right) (X_t (x,v)) \]
Using regularity of the vectorfield $\bR^3 \ni z \to z/\sqrt{1+z^2} \in \bR^3$, the proof of (\ref{eq:propreg}) can be easily reduced to the non-relativistic case handled in \cite[Prop. B.1]{BPSS} for $k \leq 5$. The arguments can be easily extended to the case $k=6$.
\end{proof}

Finally, we also need to propagate some commutator bounds along Hartree dynamics. 
We proceed here similarly as in \cite[Proposition C.1]{BPSS}. 
\begin{proposition}
Assume
\begin{equation}\label{eq:prop-assV}
\int d p \; |\widehat{V}(p)| (1+|p|^2) < \infty.
\end{equation}
Let $\omega_{N,t}$ be the solution of the relativistic Hartree equation
\[ i \varepsilon \partial_t \omega_{N,t}= \left[ \sqrt{1-\varepsilon^2\Delta} + \left( V * \rho_t\right), \omega_{N,t} \right] \]
with initial data $\omega_{N,t=0}=\omega_N$. Then there exists a constant $C>0$, such that
\[ \begin{split}  
\| [x, \omega_{N,t}]\|_{\mathrm{HS}} &\leq C e^{C |t|} \left( \| [x, \omega_{N}] \|_{\mathrm{HS}} +  \| [\varepsilon\nabla, \omega_{N}]\|_{\mathrm{HS}}\right) \\
\| [\varepsilon\nabla, \omega_{N,t}]\|_{\mathrm{HS}} &\leq C e^{C |t|} \left( \| [x, \omega_{N}]\|_{\mathrm{HS}} +  \| [\varepsilon\nabla, \omega_{N}]\|_{\mathrm{HS}}\right).
\end{split} \]
Moreover
\[ \begin{split} 
\| [x, \omega_{N,t}]\|_{\mathrm{tr}} &\leq C e^{C |t|} \left( \| [x, \omega_{N}]\|_{\mathrm{tr}} +  \| [\varepsilon\nabla, \omega_{N}]\|_{\mathrm{tr}}\right)\\
\| [\varepsilon\nabla, \omega_{N,t}]\|_{\mathrm{tr}} &\leq C e^{C |t|} \left( \| [x, \omega_{N}]\|_{\mathrm{tr}} +  \| [\varepsilon\nabla, \omega_{N}]\|_{\mathrm{tr}}\right).
\end{split} \]
\end{proposition}

\begin{proof}
Let $h_H(t)= \sqrt{1-\varepsilon^2\Delta} + \left( V * \rho_t\right) (x)$ and $\mathcal{U}(t;s)$ be the unitary evolution generated by $h_H(t)$, i.e. 
\begin{align*}
i \varepsilon \partial_t \mathcal{U}(t;s) = h_H(t) \mathcal{U}(t;s),
\end{align*}
with $\mathcal{U}(s;s)=1$ for all $s \in \mathbb{R}$. We compute 
\[ \begin{split} 
	i \varepsilon \partial_t \mathcal{U}^*(t;0)[x,\omega_{N,t}]\mathcal{U}(t;0) = \; & - \mathcal{U}^*(t;0) \left[ h_H(t), [x, \omega_{N,t}]\right] \mathcal{U}(t;0) + \mathcal{U}^*(t;0) \left[ x, [h_H(t), \omega_{N,t}] \right] \mathcal{U}(t;0)\\
= \; & \mathcal{U}^*(t;0) \left[ [h_H(t), x], \omega_{N,t}\right] \mathcal{U}(t;0)\\
= \; & \varepsilon \, \mathcal{U}^*(t;0) \left[ \frac{i \varepsilon\nabla}{\sqrt{1-\varepsilon^2\Delta}}, \omega_{N,t}\right] \mathcal{U}(t;0),
\end{split} \]
where we used the Jacobi's identity. We have 
\begin{equation}\label{eq:prop1} \begin{split} 
i \varepsilon \partial_t& \mathcal{U}^*(t;0)[x,\omega_{N,t}]\mathcal{U}(t;0)\\ 
= \; & \varepsilon \, \mathcal{U}^*(t;0) i \varepsilon\nabla \left[ \frac{1}{\sqrt{1-\varepsilon^2\Delta}}, \omega_{N,t}\right] \mathcal{U}(t;0)+ \varepsilon \, \mathcal{U}^*(t;0)  \left[ i \varepsilon\nabla, \omega_{N,t}\right]\frac{1}{\sqrt{1-\varepsilon^2\Delta}} \mathcal{U}(t;0).
\end{split} \end{equation}
We write, for an appropriate constant $C \in \bR$,  
\[ \begin{split} 
\left[ \frac{1}{\sqrt{1-\varepsilon^2\Delta}}, \omega_{N,t}\right] = \; &C  \int_0^\infty \frac{d \nu}{\sqrt{\nu}} \left[ \frac{1}{\nu+1-\varepsilon^2\Delta}, \omega_{N,t}\right]\\
= \; &C  \int_0^\infty \frac{d \nu}{\sqrt{\nu}}\frac{1}{\nu+1-\varepsilon^2\Delta} \left[ -\varepsilon^2 \Delta, \omega_{N,t}\right]\frac{1}{\nu+1-\varepsilon^2\Delta} \\
= \; &C  \int_0^\infty \frac{d \nu}{\sqrt{\nu}}\frac{i \varepsilon\nabla}{\nu+1-\varepsilon^2\Delta} \left[ i \varepsilon \nabla, \omega_{N,t}\right]\frac{1}{\nu+1-\varepsilon^2\Delta}  \\
	&+C \int_0^\infty \frac{d \nu}{\sqrt{\nu}}\frac{1}{\nu+1-\varepsilon^2\Delta} \left[ i \varepsilon\nabla, \omega_{N,t}\right]\frac{i \varepsilon\nabla}{\nu+1-\varepsilon^2\Delta}.
\end{split}\]
Inserting in (\ref{eq:prop1}) and integrating over time, we find 
\[ \begin{split}
	[x,\omega_{N,t}]
	= \; & \mathcal{U}(t;0) \, [x,\omega_{N}] \, \mathcal{U}^*(t;0)\\
	&+i C \int_0^t d \tau \int_0^\infty \frac{d \nu}{\sqrt{\nu}}  \,  \mathcal{U}(t; \tau) \frac{i \varepsilon\nabla}{\nu+1-\varepsilon^2\Delta} \left[ i \varepsilon \nabla, \omega_{N,\tau}\right]\frac{1}{\nu+1-\varepsilon^2\Delta}\, \mathcal{U}^*(t; \tau)  \\
	&+i C \int_0^t d \tau \int_0^\infty \frac{d \nu}{\sqrt{\nu}}  \, \mathcal{U}(t; \tau) \frac{1}{\nu+1-\varepsilon^2\Delta} \left[ i \varepsilon\nabla, \omega_{N,\tau}\right]\frac{i \varepsilon\nabla}{\nu+1-\varepsilon^2\Delta}\, \mathcal{U}^*(t;\tau)\\
	&+i \int_0^t d \tau \; \mathcal{U}^*(t;\tau)  \left[ i \varepsilon\nabla, \omega_{N,\tau}\right]\frac{1}{\sqrt{1-\varepsilon^2\Delta}} \, \mathcal{U}(t;\tau).
\end{split} \]
Hence 
\begin{equation}\label{app:bound1}
\begin{split} 
\| [x,\omega_{N,t}]\|_{\mathrm{HS}} \leq \; & \|[x,\omega_{N,t}]\|_{\mathrm{HS}} + C \int_0^t d \tau  \int_0^\infty \frac{d \nu}{\sqrt{\nu}} \frac{1}{\nu+1} \, \| [ \varepsilon \nabla, \omega_{N,\tau}] \|_{\mathrm{HS}} + \int_0^t d \tau \; \|[ \varepsilon\nabla, \omega_{N,\tau}]\|_{\mathrm{HS}} \\
\leq \; & \|[x,\omega_{N,t}]\|_{\mathrm{HS}}+ C \int_0^t d \tau \; \| [ \varepsilon \nabla, \omega_{N,\tau}] \|_{\mathrm{HS}}. 
\end{split}\end{equation}
Next, we observe that 
\[ \begin{split}
i \varepsilon \partial_t \mathcal{U}^*(t;0)[\varepsilon\nabla,\omega_{N,t}]\mathcal{U}(t;0) = \; & - \mathcal{U}^*(t;0) \left[ h_H(t), [\varepsilon\nabla, \omega_{N,t}]\right] \mathcal{U}(t;0) + \mathcal{U}^*(t;0) \left[ \varepsilon \nabla, [h_H(t), \omega_{N,t}] \right] \mathcal{U}(t;0)\\
= \; & \mathcal{U}^*(t;0) \left[ [h_H(t), \varepsilon\nabla], \omega_{N,t}\right] \mathcal{U}(t;0)\\
= \; & \varepsilon \, \mathcal{U}^*(t;0) \left[ \nabla \left(V * \rho_t\right), \omega_{N,t}\right] \mathcal{U}(t;0)
\end{split} \] 
Expanding the interaction in a Fourier integral, using $|\widehat{\rho}_t (k)| \leq \| \rho_t \|_1 = 1$ and the assumption (\ref{eq:prop-assV}), we obtain 
\begin{equation}\label{ap:bound2}
\begin{split}
\| [\varepsilon\nabla,\omega_{N,t}]\|_{\mathrm{HS}} \leq \; & \| [\varepsilon\nabla, \omega_N]\|_{\mathrm{HS}} + \int d k \; |k|^2 |\widehat{V}(k)|  |\widehat{\rho}(k)| \int_0^t d \tau \;  \| [ x, \omega_{N,\tau}] \|_{\mathrm{HS}} \\
\leq \; & \| [\varepsilon\nabla, \omega_N]\|_{\mathrm{HS}} + C \int_0^t d \tau \; \| [ x, \omega_{N,\tau}] \|_{\mathrm{HS}}.
\end{split} \end{equation}
Combining now \eqref{app:bound1} and \eqref{ap:bound2}, we find, applying Gronwall's lemma,
\begin{align*}
	&\| [x, \omega_{N,t}]\|_{\mathrm{HS}} + \| [\varepsilon\nabla, \omega_{N,t}]\|_{\mathrm{HS}} \leq C e^{C |t|} \left( \| [x, \omega_{N}]\|_{\mathrm{HS}} +  \| [\varepsilon\nabla, \omega_{N}]\|_{\mathrm{HS}}\right).
\end{align*}
The estimates for the trace norms of the commutators can be shown in the same way. 
\end{proof}


\begin{thebibliography}{10}

\bibitem{AKN} L. Amour, M. Khodja and J. Nourrigat.
The semiclassical limit of the time dependent Hartree-Fock equation: the Weyl symbol of the solution.
\emph{Anal. PDE} {\bf 6} (2013), no. 7, 1649--1674. 

\bibitem{AKN2} L. Amour, M. Khodja and J. Nourrigat.
The classical limit of the Heisenberg and time dependent Hartree-Fock equations: the Wick symbol of the solution. {\it Math. Res. Lett.} {\bf 20} (2013), no. 1, 119--139.

\bibitem{APPP} A. Athanassoulis, T. Paul, F. Pezzotti and M. Pulvirenti. Strong Semiclassical Approximation of Wigner Functions for the Hartree Dynamics. \emph{Rend. Lincei Mat. Appl.} \textbf{22}  (2011), 525--552.


\bibitem{BBPPT} V. Bach, S. Breteaux, S. Petrat, P. Pickl, T. Tzaneteas. Kinetic energy estimates for the accuracy of the time-dependent Hartree-Fock approximation with Coulomb interaction. Preprint arXiv:1403.1488.
  
\bibitem{BJPSS} N.~{Benedikter}, V.~Jaksic, M.~{Porta}, C. Saffirio and B.~{Schlein}. Mean-field Evolution of Fermionic Mixed States. {\it Comm. Pure Appl. Math.} {\bf 69} (2016), no. 12, 2250–-2303. 

\bibitem{BPSS} N.~{Benedikter}, M.~{Porta}, C. Saffirio and B.~{Schlein}. From the Hartree dynamics to the Vlasov equation. {\it Arch. Rat. Mech. Anal.} 

\bibitem{BPS1} N.~{Benedikter}, M.~{Porta} and B.~{Schlein}. {Mean-field evolution of fermionic systems}. \emph{Comm. Math. Phys.} {\bf 331} (2014), 1087--1131.

\bibitem{BPS2} N.~{Benedikter}, M.~{Porta} and B.~{Schlein}. {Mean-field dynamics of fermions with relativistic dispersion}. \emph{J. Math. Phys.} {\bf 55} (2014), 021901.
 
 
\bibitem{D} R.~L.~{Dobrushin}. {Vlasov equations}. \emph{Functional Analysis and Its Applications.} {\bf 13} (1979), no.~2, 115--123.


\bibitem{EESY}
A.~{Elgart}, L.~{Erd{\H{o}}s}, B.~{Schlein} and H.-T.~{Yau}. {Nonlinear {H}artree equation as the mean field limit of weakly
  coupled fermions}. \emph{J. Math. Pures Appl. (9)} \textbf{83} (2004), no.~10,
  1241--1273.
  
  
\bibitem{LP}
P.-L. Lions and T. Paul. Sur les mesures de Wigner. {\it Rev. 
Mat. Iberoamericana} {\bf 9} (1993), 553--618. 


\bibitem{NS}
H.~{Narnhofer} and G.~L.~{Sewell}. {Vlasov hydrodynamics of a quantum
  mechanical model}. \emph{Comm. Math. Phys.} \textbf{79} (1981), no.~1, 9--24.
 

\bibitem{PP} S.~{Petrat} and P.~{Pickl}. {A new method and a new scaling for deriving fermionic mean-field dynamics}. {\tt arXiv:1409.0480}.

\bibitem{PP09} F. Pezzotti and M. Pulvirenti. Mean-field limit and Semiclassical Expansion of a Quantum Particle System. {\it Ann. H. Poincar\'e} {\bf 10} (2009), no. 1, 145--187.

\bibitem{Sp}
H.~{Spohn}. {On the {V}lasov hierarchy}, \emph{Math. Methods Appl. Sci.} \textbf{3}
  (1981), no.~4, 445--455.
\end{thebibliography}
\end{document}